\documentclass[11pt]{article}

\usepackage[english]{babel}
\usepackage[utf8]{inputenc}
\usepackage[T1]{fontenc}
\usepackage[a4paper, margin=1in]{geometry}
\usepackage{graphicx}
\usepackage{todonotes}
\usepackage{color}
\definecolor{darkgreen}{rgb}{0,0.5,0}
\usepackage{hyperref}
\hypersetup{
    colorlinks=true,        
    linkcolor=red,          
    citecolor=darkgreen,        
    filecolor=magenta,      
    urlcolor=cyan           
}

\usepackage{amsmath}
\usepackage{amsthm} 
\usepackage{amssymb}
\usepackage{amsfonts}
\usepackage{mathrsfs}
\usepackage{mathtools}

\usepackage{verbatim}
\usepackage{footnote}
\usepackage[linesnumbered,boxed, vlined]{algorithm2e}
	\DontPrintSemicolon
  \SetKwInOut{Input}{Input}
  \SetKwInOut{Output}{Output}
\usepackage{lineno}
\usepackage{caption}
\usepackage{framed}
\usepackage{enumerate}
\usepackage{paralist}
\usepackage{tikzsymbols}
\usepackage{thmtools,thm-restate}
\usepackage{nicefrac}
\usepackage{stmaryrd}

\usepackage{subcaption}

\usepackage[capitalize, nameinlink]{cleveref}
\usepackage[section]{placeins}
\crefname{theorem}{Theorem}{Theorems}
\Crefname{lemma}{Lemma}{Lemmas}
\Crefname{invariant}{Invariant}{Invariants}
\Crefname{claim}{Claim}{Claims}
\Crefname{observation}{Observation}{Observations}
\Crefname{algorithm}{Algorithm}{Algorithms}
\Crefname{figure}{Figure}{Figures}

\DeclareMathOperator*{\poly}{poly}

\newtheorem{theorem}{Theorem}
\newtheorem{lemma}[theorem]{Lemma}
\newtheorem{corollary}[theorem]{Corollary}

\newtheorem{definition}[theorem]{Definition}

\newtheorem{example}[theorem]{Example}
\newtheorem{observation}[theorem]{Observation}
\newtheorem{claim}[theorem]{Claim}
\newtheorem{remark}{Remark}
\newtheorem*{remark*}{Remark}

\usepackage{thmtools,thm-restate}

\newcommand{\LOCAL}{\ensuremath{\mathsf{LOCAL}}\xspace}
\newcommand{\CONGEST}{\ensuremath{\mathsf{CONGEST}}\xspace}

\newcommand{\E}{\mathbb{E}}
\newcommand{\eps}{\varepsilon}

\newcommand{\lovasz}{Lov\'{a}sz\xspace}

\newcommand{\Bad}{\textsf{Bad}}
\newcommand{\pres}{\mathsf{pre}}
\newcommand{\post}{\mathsf{post}}

\newcommand{\Lmark}{\mathsf{L}}
\newcommand{\Rmark}{\mathsf{R}}
\newcommand{\Bmark}{\mathsf{B}}
\newcommand{\VL}{V^\Lmark}
\newcommand{\VR}{V^\Rmark}
\newcommand{\DeltaL}{\Delta_\Lmark}
\newcommand{\DeltaR}{\Delta_\Rmark}

\newcounter{constant}

\newcommand{\cte}{c_{\arabic{constant}}}
\newcommand{\inccte}{\stepcounter{constant}\cte}

\newcommand{\Var}{\textsf{Var}}
\newcommand{\Event}{\ensuremath{\mathcal{X}}}

\newcommand{\event}{X}

\newcommand{\cL}{{\cal L}}

\newcommand{\cC}{{\cal C}}
\newcommand{\cB}{{\cal B}}

\newcommand{\frozen}{\ensuremath{\mathtt{frozen}}\xspace}

\hypersetup{
	colorlinks=true,
	linkcolor=black,
	citecolor=black,
	filecolor=black,
	urlcolor=[rgb]{0,0.1,0.5},
	pdftitle={ Fast Distributed Vertex Splitting With Applications},
	pdfauthor={Magnus M. Halldorsson, Yannic Maus and Alexandre Nolin}
}

\DeclareMathOperator*{\Exp}{\mathbb{E}}

\DeclarePairedDelimiter{\card}{\lvert}{\rvert}
\DeclarePairedDelimiter{\parens}{\lparen}{\rparen}
\DeclarePairedDelimiter{\ceil}{\lceil}{\rceil}

\newcommand{\myemail}[1]{\,$\cdot$\, {\small #1}}
\newcommand{\myaff}[1]{\,$\cdot$\, {\small #1}\par\medskip}

\newenvironment{myabstract}
{\list{}{\listparindent 1.5em%
		\itemindent    \listparindent
		\leftmargin    1cm
		\rightmargin   1cm
		\parsep        0pt}%
	\item\relax}
{\endlist}

\newenvironment{mycover}
{\list{}{\listparindent 0pt
		\itemindent    \listparindent
		\leftmargin    1cm
		\rightmargin   1cm
		\parsep        0pt}%
	\raggedright
	\item\relax}
{\endlist}

\begin{document}
\begin{mycover}
	{\huge\bfseries\boldmath Fast Distributed Vertex Splitting \newline with Applications \par}
	\bigskip
	\bigskip
	\bigskip
	
	\textbf{Magn\'us M. Halld\'orsson}
	\myemail{mmh@ru.is}
	\myaff{Reykjavik University, Iceland}
	
	\textbf{Yannic Maus}
	\myemail{yannic.maus@ist.tugraz.at}
	\myaff{TU Graz, Austria}
	
	\textbf{Alexandre Nolin}
	\myemail{alexandren@ru.is}
	\myaff{Reykjavik University, Iceland}
\end{mycover}

\begin{myabstract}
	\noindent\textbf{Abstract.}
		We present $\poly\log\log n$-round randomized distributed algorithms to compute vertex splittings, a partition of the vertices of a graph into $k$ parts such that a node of degree $d(u)$ has $\approx d(u)/k$ neighbors in each part. Our techniques can be seen as the first progress towards general $\poly\log\log n$-round algorithms for the \lovasz Local Lemma.  
	
	As the main application of our result, we obtain a randomized $\poly\log\log n$-round \CONGEST algorithm for $(1+\eps)\Delta$-edge coloring $n$-node graphs of sufficiently large constant maximum degree $\Delta$, for any $\eps>0$. 
	Further, our results improve the computation of defective colorings and certain tight list coloring problems. 
	All the results improve the state-of-the-art round complexity exponentially, even in the \LOCAL model. 
\end{myabstract}

\thispagestyle{empty}
\setcounter{page}{0}
\newpage

\thispagestyle{empty}
\tableofcontents
\clearpage
\section{Introduction}
Consider the following fundamental load-balancing problem: Partition the vertices of an $n$-node degree-$\Delta$ graph into two parts so that each node has at most $(1+\eps)\Delta/2$ neighbors in each part, where $\eps > 0$ is an arbitrary given constant.
When $\Delta$ is large enough (say, superlogarithmic), such a \emph{2-splitting} is trivially achieved w.h.p.\ without communication. Can it be solved fast distributively for arbitrary $\Delta$?

The 2-splitting problem can be formulated as an instance of the \emph{\lovasz local lemma}~(LLL). Consider some \emph{``bad''} events over a probability space. The celebrated \lovasz local lemma
states that if the events satisfy certain limited dependencies, then there is a positive probability that none of them happens \cite{LLL73}. In the $2$-splitting problem, the probability space is spanned by each node picking a part uniformly at random and there is a bad event for each node that occurs when too many of its neighbors are in any one of the parts. 
In the constructive version of the LLL, the objective is to also compute an assignment avoiding all bad events, and using known distributed LLL algorithms, it can be solved in $O(\log n)$ distributed rounds \cite{MoserTardos10,CPS14}. 
For small $\Delta$, there is a faster $O(\Delta^2 + \poly \log\log n)$-round algorithm \cite{FG17}, 
but it does not improve the case of arbitrary $\Delta$. 
This leaves a major open problem: Can we close the gap between the $O(\log n)$ upper bound and $\Omega(\log_{\Delta}\log n)$ lower bound \cite{BGKMU19}?
Clarifying this for 2-splitting would be the first step towards resolving the complexity of general distributed LLLs.

Our central technical result is to answer this question affirmatively, even in the bandwidth-constrained \CONGEST model, by giving a $\poly\log\log n$-round algorithms for 2-splitting and various other vertex- and edge splitting problems.
Note that these are also exponential improvements for the \LOCAL model.
Such splitting problems are pervasive in distributed graph algorithmics \cite{SLOCAL17,GS17,F17,FGK17,GHKMSU17,GHK17,HarrisEdge19,BGKMU19}. They can be viewed as questions of rounding and discrepancy, and they are frequently the major building block in solving various classic problems when using a divide-and-conquer approach.


We illustrate the reach of the techniques by giving much faster algorithms for two classic coloring problems.

\begin{restatable}[Edge coloring]{theorem}{thmEdgeColoring}
	\label{thm:edgeColoring}
	For any constant $\eps>0$, 
	there is a $\poly\log\log n$-round randomized algorithm to compute a $(1+\eps)\Delta$-edge coloring on any graph with maximum degree $\Delta\geq \Delta_0$ where $\Delta_0$ is a sufficiently large constant.
\end{restatable}
Notice that as a function of $n$ alone, previous methods use at least $\Omega(\log n)$ time, even in the \LOCAL model. 
The problem has a $\Omega(\log_{\Delta}\log n)$ lower bound \cite{CHLPU19}. Previously, 
$\poly\log\log n$-round algorithms were only known for $2\Delta-1$-edge coloring \cite{EPS15,representativesets}, while algorithms using any smaller number of colors were $O(\log_\Delta n + \poly \log \log n)$-round \cite{DGP98,PanconesiS97,CPS14,EPS15,CHLPU19}, even in the \LOCAL model. 
Tackling this problem in \CONGEST is non-trivial as it depends on LLL, which only has efficient known \CONGEST solutions for the constant-degree case \cite{MU21}.

\smallskip

In the second application, the \emph{$(L,T)$-list coloring} problem, each node of a graph is given a list of at least $L$ colors such that any color in its list appears in at most $T$ neighbors' lists. We ask for a valid node coloring where each node receives a color from its list, with the ratio $L/T$ as small as possible. Observe that the degree of a node can be much larger than its list of colors, and thus greedy approaches are insufficient, even centrally.

\begin{restatable}[List coloring]{theorem}{thmListColoring}
	\label{thm:ListColoring}
	There is a $\poly \log\log n$-round randomized \LOCAL  algorithm for the list coloring problem, for any $T$ and $L$ with $L \ge (1+\delta)T$, for any $\delta > 0$ and any $\Delta \ge \Delta_0$, for some absolute constant $\Delta_0$.
\end{restatable}

Previous algorithms either used $O(\log n)$-rounds \cite{CPS14} or required $L/T\geq C_0$ for a (large) constant $C_0$ \cite{FG17}. See \Cref{sec:listColoring} for more related work on list coloring.

\subsection{Contributions on Splitting problems}

The main ingredient for both of the above applications is our efficient method to split graphs into small degree subgraphs.
A  \emph{$k$-vertex splitting problem} with \emph{discrepancy} $z$ is a partition of the vertex set into $k$ parts $V_1,\ldots,V_k$ such that, for each $i\in[k]$, each node $v\in V$ has $d(v)/k\pm z$ neighbors in $V_i$.  
Intuitively, splitting a graph into $k$ parts with a discrepancy of $\eps \Delta/k$ is useful to solve various problems that are easier on low-degree graphs.
These problems must be resilient to imperfect splits, which is ensured in coloring problems by having a surplus of colors.

A Chernoff bound argument shows that such splittings are quite easy for high degrees ($\Delta \gg k \log n$).
We obtain the following theorem:

\begin{restatable}{theorem}{thmVertexSplitting} 
	\label{thm:vertexSplitting}
	There exists a universal constant $\inccte>0$ s.t.:
	For any $\eps>0$, maximum degree $\Delta\leq \poly\log n$, and 
	$k \leq \cte \cdot (\eps^{4} \Delta/\ln \Delta)$,
	there is a distributed randomized \LOCAL algorithm to compute a $k$-vertex splitting with discrepancy $\eps\Delta/k$ in  $O(1/\eps)+\poly \log\log n$ rounds.
\end{restatable}
The $\poly\log\log n$ term in the runtime of \Cref{thm:vertexSplitting} stems from solving LLL instances of size $N=\poly\log n$ deterministically. Any improvement on such algorithms immediately carries over to our result. However, there is a lower bound of $\Omega(\log_\Delta\log n)$ rounds for randomized and $\Omega(\log_{\Delta}n)$ rounds for deterministic algorithms for the respective splitting problems (and hence also for the LLL problem) \cite{BGKMU19}. These lower bounds even hold for a weak variant of the vertex splitting problem, in which each node only needs to have one neighbor of each color class.

\subparagraph{Variants.} 
Our main applications require subtle variations of the splitting problem. To this end, we solve a more general problem, where we separate the two functions of each node: as a variable (which part is it assigned to) and as an event (whether its neighborhood is evenly split).
In the \emph{bipartite $k$-vertex splitting problem with discrepancy $z$}, we have a set $\VL$ of nodes for the events
and a set $\VR$ of nodes for the variables, with an edge between every dependent variable-edge pair. We wish to partition $\VR$ into $k$ parts such that each event vertex $u\in \VL$ has $d(u)/k \pm z$ neighbors in each part.

\begin{restatable}{theorem}{thmbipartiteVertexSplitting}
	\label{thm:bipartiteSplitting}
	There exists a universal constant $\inccte>0$ s.t.:
	For any $\eps>0$, maximum degree $\Delta\leq \poly\log n$ and 
	$k\leq \cte \cdot (\eps^{4}\DeltaL/\ln \Delta)$,
	there is a distributed randomized \LOCAL algorithm to compute a bipartite $k$-vertex splitting problem with discrepancy  $\eps\DeltaL/k$ in  $O(1/\eps)+ \poly\log \log n$ rounds.
\end{restatable}

We also devise \CONGEST versions of \Cref{thm:bipartiteSplitting,thm:vertexSplitting} that are essential to our edge coloring result in \CONGEST. The formal statement appears in \Cref{thm:CONGESTsplittingBoth} and requires $k$ to be a $O(\log^2\log n)$ factor smaller than in \Cref{thm:bipartiteSplitting,thm:vertexSplitting}.

In a \emph{$d$-defective $c$-coloring}, each of the $c$ color classes induces a graph of maximum degree $d$. Defective colorings are frequently used in divide-and-conquer approaches to other coloring problems \cite{Kuhn2009WeakColoring,BE09} and they have been studied in several works, e.g., \cite{Kuhn2009WeakColoring,BE09,KS18,GK20,M21}, usually stating variations of deterministic algorithms for computing $d$-defective coloring with $O((\Delta/d)^2)$ colors. 
As any vertex splitting is also a defective coloring, \Cref{thm:vertexSplitting}  implies a $\poly\log\log n$-round algorithm for  $(1+\eps)\Delta/k$-defective $k$-coloring. Previous algorithms for $(1+\eps)\Delta/k$ defective colorings either used $O(k^2)$ colors \cite{Kuhn2009WeakColoring,M21} or a logarithmic number of rounds through solving the respective LLL problem \cite{CPS14}.

\subsection{Challenges to Fast and Efficient Splitting}
Known approaches to splitting (or any of the other problems we consider) all build on the 
\emph{\lovasz Local Lemma} (LLL) for the low-degree case ($\Delta < \log n$).
This
hits a wall, since there are no 
strongly sublogarithmic time distributed LLL algorithms known, in spite of intensive efforts \cite{CP19,CHLPU19}.

There are two known approaches to distributed LLL algorithms. 
The breakthrough Moser-Tardos method \cite{MoserTardos10,CPS14}  is based on stochastic local search,
which appears to inherently require logarithmic rounds. The other approach is to use the \emph{shattering} technique, solving most of the problem quickly, leading to small remaining ``shattered'' subgraphs for which we can afford to apply slower techniques. This was introduced by Beck \cite{Beck91} in the centralized setting and Alon \cite{Alon91LLL} in the parallel setting.

Fischer and Ghaffari \cite{FG17} proposed a shattering-based distributed algorithm, modeled on an earlier sequential algorithm of Molloy and Reed \cite{MR98LLL}.
Using recent network decompositions \cite{RG19}, their method runs in $O(\Delta^2 + \poly \log\log n)$ time, which is fast for low-degree graphs ($\Delta \le \poly\log\log n$) but doesn't improve the general case.
To understand the issue, let us examine more closely the reasoning behind the method of \cite{FG17} in the context of 2-splitting.

A random assignment (of the nodes into the parts) is easily seen to satisfy the lion's share of the vertices, where ``satisfied'' means having discrepancy within the stated bound.
Each node is so likely to be satisfied that the remaining subgraph is indeed \emph{shattered}: the connected components induced by the set of unsatisfied nodes are of small size (assuming $\Delta \le \poly\log n$, which is the hard case). 
One natural approach is to \emph{undo} the assignment to the unsatisfied nodes, and then solve the problem separately on the unsatisfied nodes. However, this causes new problems: Nodes that were previously satisfied may become 
hard to satisfy.
For instance, suppose a node $v$ has neighbors $u_1, \ldots, u_t$ assigned to the first part and nodes $u_{t+1}, \ldots, u_{2t}$ assigned to the second part, for a perfect split. But it is now possible that all of $u_1, u_2, \ldots, u_t$ are retracted, having their assignment undone. 
Then, satisfying $v$ now requires assigning its neighbors back to the second part, which leaves little flexibility, and there may not be a valid solution. 

Fischer and Ghaffari \cite{FG17} (and \cite{MR98LLL}) fix this by sampling the random variables (i.e., which part each node is assigned to) only \emph{gradually}, i.e., at most one variable per event is sampled simultaneously.  Along with ``freezing'' (or deferring) certain nodes, this ensures that no vertex experiences too heavy a setback caused by retractions.
The gradual sampling is achieved by first computing a distance-2 coloring of the graph using $O(\Delta^2)$ colors, and then sampling only the nodes of a single color class at a time.
The downside is that this unavoidably requires time complexity at least $\Delta^2$. 

A different type of challenge appears when aiming for bandwidth-efficient algorithms. Even if one drastically improves upon the sketched $O(\Delta^2)$ ``pre-shattering'' procedure from \cite{FG17}, the deterministic procedure used in the ``post-shattering'' phase of their algorithm to complete the obtained partial solution makes heavy use of the unlimited bandwidth of the \LOCAL model. In fact, while both types of randomized distributed LLL methodologies \cite{MT20,CPS14} are themselves frugal in terms of bandwidth, 
known deterministic LLL algorithms are based on bandwidth-hungry generic derandomization results \cite{SLOCAL17,GHK17,RG19}.

\subsection{Our Methods in a Nutshell}

\subparagraph{Fast splitting.}
Our approach is to sample gradually -- like in \cite{MR98LLL,FG17} -- but faster.
We group the variables (representing the part assigned to a node) into \emph{buckets} and then sample the variables one bucket at a time. This is crucially done so that the impact of any given bucket on any given event is limited (namely, the number of neighbors of a node in any given bucket is upper bounded), so that we can recover from bad probabilistic assignments. Intuitively, a node might have to ``give up'' on all of its neighbors inside a bucket, i.e., it may be that their assignment is chosen adversarially.  As we can guarantee that each event has to give up on at most one bucket, it turns out to suffice to use a constant number of buckets to get a good split, or more generally $O(1/\eps)$ buckets to get $(1+\epsilon)$-approximate split.
Generating this bucket assignment is itself a splitting problem (that we term a \emph{$q$-divide}) requiring the use of LLL, but 
one with less moving parts and a much simpler analysis in \LOCAL. In the \CONGEST model, it still requires a novel post-shattering phase.

\subparagraph{Post-shattering in \CONGEST.} 
We solve the post-shattering phase as a sequence of successive relaxations, one for each disjoint group of clusters of the network decomposition. In effect, we solve a new LLL for each cluster group, with progressively stricter criteria (due to the accumulated discrepancy). Each relaxation is solved by a randomized, rather than a deterministic, algorithm. Namely, we run $O(\log n)$ independent instances of the Moser-Tardos process on the cluster, and since each succeeds with constant probability, we achieve at least one valid solution, w.h.p. 
This parallel instance technique was introduced by Ghaffari \cite{Ghaffari2019} for problems like $\Delta+1$-coloring, a simpler setting where the problem is always solvable on clusters processed \emph{later}, regardless of how the \emph{earlier} clusters are solved.

In our edge-coloring application, we use splitting to whittle down the degree parameter to a manageable size. Once degrees are down to $\poly\log\log n$, we can simulate the known algorithms from the \LOCAL model, including derandomization techniques, to solve them also in 
$\poly\log\log n$ \CONGEST rounds. 

For our list coloring results, we use our splitting procedure to first reduce the parameter $L$ and $T$ to $\poly\log\log n$  while keeping the initial ratio $L/T\geq (1+\delta)$ almost intact. Then, in additional color pruning steps we amplify the ratio until it is larger than a sufficiently large constant, at which point the problem can be solved efficiently via a known LLL-based method \cite{FG17}.

\subsection{Further Related Work}
The only known LLL algorithm in the \CONGEST model is by Maus and Uitto \cite{MU21} who provide a $\poly\log\log n$-round algorithm for LLLs with a polynomial criterion and constant dependency degree. In this work, we observe that their runtime remains $\poly\log\log n$, even if the dependency degree is as large as $\poly\log\log n$, see \Cref{lem:CONGESTpostshatteringSmallDegree} for details.

\subparagraph{Edge coloring.}
Dubhashi, Grable, and Panconesi \cite{DGP98} gave a distributed algorithm for $(1+\epsilon)\Delta$-edge coloring based on the R\"odl nibble method. Their results only apply to large values of $\Delta$.
Elkin, Pettie, and Su \cite{EPS15} extended the reach to arbitrary $\Delta$ by reduction to distributed LLL, and obtained improved complexity of $O(\log^* \Delta \cdot \ceil{ \log n / \Delta^{1-o(1)}} )$. 
Chang et al.\ \cite{CHLPU19} improved the complexity  to $O(\log_\Delta n + \log^{3+o(1)}\log n)$ for $\eps^{-1}\in O(1)$, and to $O(\log n)$ for $\eps^{-1} \in \tilde{O}(\sqrt{\Delta})$. 

There are clear tradeoffs between the number of colors and the time complexity. Computing $(2\Delta-1)$-edge coloring can be achieved in $\poly\log\log n$ rounds \cite{EPS15} (even in {\CONGEST} \cite{representativesets}), and even in $O(\log^* n)$ rounds for $\Delta \ge \log^2 n$ \cite{EPS15,representativesets}). 
Chang et al. \cite{CHLPU19}
showed via the round elimination method that 
computing a $(2\Delta-2)$-edge coloring requires $\Omega(\log_\Delta\log n)$ rounds. A
$\poly(\Delta,\log n)$-round algorithm is known for $\Delta+2$-coloring \cite{SuVu19} 
and very recently for $\Delta+1$-coloring \cite{Bernshteyn22}. Chang et al.\ \cite{CHLPU19} showed that an (possibly randomized) algorithm for $\Delta+1$-coloring that can start with any partial coloring requires $\Omega(\Delta \log n)$ rounds. 
They also showed that $(1+\log \Delta/\sqrt{\log \Delta})\Delta$-edge coloring can be found in $O(\log n)$ rounds.

\subparagraph{Splitting.}
Ghaffari and Su~\cite{GS17} gave three \LOCAL algorithms for splitting the edges of a graph into two parts such that each node has at most $(1+\eps)\Delta/2$ incident edges in each part, rounded up for their randomized result.

Their deterministic algorithms achieve complexity $O(\eps^{-1}\Delta^2\log^5 n)$ when $\Delta \geq c\cdot \eps^{-1}\log n$, and complexity $O(\eps^{-3}\log^7 n)$ when $\Delta \geq c \cdot \eps^{-2}\log n$, where $c$ is a suitable absolute constant.
Their randomized algorithm solves the problem for all $\Delta$ in $O(\eps^{-2}\Delta^2\log^4 n)$ rounds.
These results were later improved by \cite{GHKMSU17} to $O(\eps^{-1-o(1)}\log n)$ rounds for deterministic algorithms and   $O(\eps^{-1-o(1)}\log\log n)$ for randomized algorithms, with stronger guarantees on the split. 
However, it is unclear whether these edge-splitting algorithms can be extended to the \CONGEST model, as the algorithms communicate simultaneously over various long paths in the network.
The importance of splitting problems for the area was highlighted in \cite{SLOCAL17} and \cite{BGKMU19}. The latter gave various direct reductions of the maximal independent set problem and coloring problems to splitting problems. In addition, they studied several weak variants of the splitting problem, e.g., splitting into two parts such that each node needs to have at least one neighbors in each part. They show that even these have a $\Omega(\log_{\Delta}\log n)$ lower bound for randomized algorithms and $\Omega(\log_{\Delta}n)$ for deterministic algorithms. They also obtain a $\poly\log\log n$-round algorithm for the weak variant in the special case of regular graphs. 

\subsection{Outline}
In \Cref{sec:prelim}, we define the models, the setup for the \lovasz Local Lemma and introduce notation. 
\Cref{sec:qdivide} contains our algorithm for the $q$-divide that does not just serve as a warm-up for our more involved splitting algorithms, but is also used as a subroutine in the latter. Our main splitting algorithm is presented in \Cref{sec:splitting} for  \LOCAL  and in \Cref{sec:splittingCongest} for \CONGEST. In \Cref{sec:edgeColoring,sec:listColoring} we present our splitting applications, that is,  edge coloring and list coloring, respectively. 
In \Cref{sec:bipartiteSplitting}, we present a bipartite version of our splitting results (\Cref{thm:bipartiteSplitting}). In \Cref{thm:qdivide}, we present a local version of the $q$-divide, in which the discrepancy of a node is relative to its own degree that may be much smaller than $\Delta$. 

\section{Models, \lovasz Local Lemma, Shattering, and Notation}
\label{sec:prelim}
\subparagraph{\LOCAL and \CONGEST model \cite{linial92,peleg00}.} In the \LOCAL model, a communication network is abstracted as an $n$-node graph with maximum degree $\Delta$. Nodes communicate in \emph{synchronous rounds}, in each of which, a node can perform arbitrary local computations and send messages of arbitrary size to each of its neighbors. Initially, each node is unaware of the network topology and at the end of the computation a node has to output its own part of the solution, e.g., the colors of its incident edges in an edge coloring problem. The main complexity measure is the number of rounds until each node has produced an output. 
The \CONGEST model is identical, with the additional restriction that messages 
contain $O(\log n)$ bits. 

\subparagraph{Distributed \lovasz Local Lemma.} There are random variables \Var\  and (bad) events \Event\ at the nodes. Each event $\event$ depends on a subset $\Var(\event)$ of the random variables.
Let $p(\event)$ denote the probability that event $\event$ occurs. 
As usual, we want to find an assignment to the variables so that none of the events occur.

We form the dependency graph $H = (\Event, E_H)$ on the events, where two events $X_1, X_2 \in \Event$ are adjacent if they depend on a common variable, i.e., if $\Var(X_1) \cap \Var(X_2) \ne \emptyset$.

In a distributed setting, we assume that each variable and each event is associated with some node of the communication graph $G$. For most LLL algorithms it is essential, that the dependency graph can be simulated efficiently in the communication network. In the \LOCAL model, one round of communication in $H$ can be simulated in $t$ rounds if the variables $\Var(\event)$ upon which the event $\event$ depends are within distance $t$ in $G$ of the node where $\event$ resides. 
Let $d$ be the maximum degree of $H$, while $\Delta$ is the max degree of $G$.

Normally, an LLL is specified in terms of a function $f$, such that $p(\event)  f(d)\le 1$. 
The original specification of Lov\'asz has $f(d) = e \cdot d$ and ensures the existence of an assignment of the variables such that all bad events are avoided. In the study of distributed LLL algorithms, the functions $d^2$ \cite{CPS14}, $c \cdot d^8$ \cite{FG17} (both \emph{polynomial criteria}), and $2^{d}$ (\emph{exponential criterion}) \cite{brandt2016LLL,BMU19,BGR20} have appeared in the literature.

\begin{example}[$k$-vertex splitting with discrepancy $6\Delta/k$ is an LLL] Let each node in the graph pick one of $k$ parts, $V_1,\ldots,V_k$, uniformly at random. Introduce a \emph{bad event} $\event_v$ for each node $v\in V$ that holds if the number of neighbors of $v$ within any one part deviates from the expected value by more than $6\Delta/k$, i.e., if $|N(v)\cap  V_i|\neq d(v)/k\pm 6\Delta/k$, for some $i \in[k]$. Formally, there is one \emph{variable} for each vertex indicating the part that the vertex joins. As the event $\event_v$ shares variables only with the events in its $2$-hop neighborhood, the dependency degree of the LLL is $d \le \Delta^2$. A Chernoff bound shows that $\Pr(\event_v)\leq \exp(-\Omega(\Delta))=\exp(-\Omega(\sqrt{d}))$. Hence, 
	this splitting problem is an LLL with exponential criterion, if $\Delta$ is above an absolute constant.
\end{example}
The constant $6$ is chosen somewhat arbitrarily in order to make the Chernoff bound-based claim simple. In the following sections, we aim at splittings with discrepancy $(1+\eps)\Delta/k$.

\subparagraph{Shattering.}
Our algorithms make use of the influential shattering technique\footnote{The technique has been used extensively for efficient algorithms for various local distributed graph problems and in particular symmetry breaking problems such as sinkless orientation \cite{GS17}, $\Delta+1$-vertex coloring \cite{CLP20}, $\Delta$-coloring \cite{GHKM18}, Maximal Independent Set \cite{ghaffari16_MIS}, Maximal Matching \cite{BEPSv3}, $(2\Delta-1)$-Edge-Coloring \cite{BEPSv3}, and also for general LLL algorithms on small degree graphs \cite{FG17}.} in which one first uses a randomized algorithm to set the values of some of the variables such that \emph{unsolved} parts of the graph induce \emph{small connected components}, which are solved in the \emph{post-shattering phase}. The following lemma shows that the remaining components are indeed small. 

\begin{lemma}[Lemma 4.1 of \cite{CLP20}]
	Consider a randomized procedure that generates a subset $\Bad \subseteq  V$ of vertices. Suppose that for each $v \in  V$, we have $\Pr[v \in  \Bad] \leq \Delta^{-3c}$, and the events $v\in \Bad$ and $u\in \Bad$ are determined by non-overlapping sets of independent random variables for nodes with distance larger than $2c$. 
	Then, w.p.\ $1-n^{-\Omega(c')}$, each connected component in $G[\Bad]$ has size at most $(c' /c)\Delta^{2c} \log_\Delta  n$.
	\label{lem:shattering}
\end{lemma}
The following standard result solves these small components efficiently. 
\begin{lemma}[\cite{RG19}]
	\label{lem:LOCALpostshattering}
	There is a deterministic \LOCAL LLL algorithm with polynomial criterion that runs in $\poly\log N$ rounds on instances of size $N$, even with 
	an ID space of size exponential in $N$.
\end{lemma}
\begin{proof}[Proof Sketch]
	The result follows with the derandomization of the distributed version of Moser-Tardos \cite{MoserTardos10} via the network decomposition by Rozhon and Ghaffari \cite{RG19}, as explained in \cite{RG19}. Note that the exponential ID space is not an obstacle in the \LOCAL model as it can be circumvented by first computing a $\Theta(T)$-distance coloring with $\poly N$ colors, e.g., by using Linial's coloring algorithm \cite{linial87} if the algorithm runs in $T$ rounds.
\end{proof}

\subparagraph{Notation and concentration bounds.}
Given a graph $G=(V,E)$ and a subset $S\subseteq V$ the induced graph $G[S]$ is the graph with vertex set $S$ that contains all edges of $E$ with both endpoints in $S$. Similarly, for an edge set $F\subseteq E$, the induced graph $G[F]$ is the graph with edge set $F$ that contains all vertices that appear in an edge of $F$.  We denote $[n]=\{0,\ldots, n-1\}$. We use the following standard concentration bounds (see, e.g., \cite{DGP98}). 

\begin{lemma}[Chernoff bounds,\cite{DGP98}]\label{lem:basicchernoff}
	Let $\{X_i\}_{i=1}^r$ be a family of independent binary random variables with $\Pr[X_i=1]=q_i$, and let $X=\sum_{i=1}^r X_i$. For any $\delta>0$, $\Pr[|X-\Exp[X]|\ge \delta\Exp[X]]\le 2\exp(-\min(\delta,\delta^2) \Exp[X]/3)$.
\end{lemma}
\begin{corollary}\label{cor:handychernoff}
	With $X$ of the same form as in \cref{lem:basicchernoff}, $\forall \mu, z$ s.t.\ $z\leq \mu$ and $\Exp[X] \leq \mu$, $\Pr[|X-\Exp[X]|\ge z]\le 2\exp(-z^2/(3\mu))$.
\end{corollary}

\section{Warm-Up: Computing \texorpdfstring{$q$}{q}-divides}
\label{sec:qdivide}
For an integer $q\geq 1$, a \emph{$q$-divide} of a graph is a partition of its vertices into $q$ parts (``buckets'') $V_1,\ldots,V_q$ such that each vertex has at most $8\Delta/q$ neighbors in each bucket. We show the following theorem.
\begin{restatable}{theorem}{thmqdivide}
	\label{thm:qdivide}
	For any $\Delta\leq \poly\log n$ and  $q\in [1,(1/6)\Delta/\ln\Delta]$, there is a \LOCAL algorithm to compute a $q$-divide in $\poly\log\log n$-rounds.
\end{restatable}

We use $q$-divide as a subroutine in our $k$-splitting algorithm of \Cref{thm:vertexSplitting}. Additionally, the techniques to compute a $q$-divide serve as a warm up for the more involved algorithm for vertex splitting. 
There are two crucial differences between a (tight) $k$-splitting and a $q$-divide: 
(1) A splitting guarantees both the minimum and maximum number of neighbors of a node inside each part, while a $q$-divide gives only an upper bound, 
(2) the upper bounds asked for by a $q$-divide are loose, i.e., we deviate by a factor $8$ from a perfect partition, while a splitting is within a $(1\pm\eps)$-factor.

A $q$-divide is guaranteed to exist by LLL when $q\in O(\Delta / \log\Delta)$.  A $q$-divide can also be defined as a $8\Delta/q$-defective $8\Delta/q$-frugal $q$-coloring, where $x$-frugal means that each color appears no more than $x$ times in each neighborhood.

\begin{remark} 
\label{rem:zeroRoundQdivide} For  $\Delta/q=\Omega(\log n)$, there is a trivial zero round \CONGEST algorithm for $q$-dividing.
Each vertex assigns itself to a bucket uniformly at random ; each vertex has $d(v)/q \leq \Delta/q$ neighbors in each bucket in expectation. By \cref{lem:basicchernoff} (Chernoff bound), for each vertex $v$ and $i\in [q]$, $\Pr[|N(v)\cap V_i| >8\Delta / q] \leq \exp(-\Delta/q) \in n^{-\Omega(1)}$. Therefore, w.h.p., no vertex has more than $8\Delta/q$ neighbors in a bucket.
\end{remark}

For smaller $\Delta/q$ we give an algorithm based on shattering (that also works for large $\Delta/q$). 

The algorithm is parameterized with a threshold parameter $z(v)$ for each vertex $v$. For \Cref{thm:qdivide} we set $z(v)=8\Delta/q$ for all nodes. In \Cref{thm:qdivideLocalGuarantees} we compute slightly different versions of $q$-divides with different choices of $z(v)$. 

\subparagraph{Algorithm.}
\textbf{Phase I:} (Pre-shattering) Each vertex picks one of the first $q/2$ buckets u.a.r. Whenever a node has more than $z(v)$ neighbors in a bucket, it deselects these, i.e., these neighbors are removed from the bucket. \textbf{Phase II:} (Post-shattering) The post-shattering instance is formed by all nodes that are not assigned to any bucket, together with their neighbors. The objective is to add each unassigned node to one of the last $q/2$ buckets, such that each node has at most $z(v)$ neighbors in each bucket. In \Cref{lem:qdividePostShatteringLLL}, we show that this problem is an LLL instance with a polynomial criterion, and in 
\Cref{lem:qdivideShattering} that it is induced by connected components of small size.   
We solve it via \Cref{lem:LOCALpostshattering} in \LOCAL.

\begin{lemma}
	\label{lem:qdivideShattering}
	For threshold discrepancy $z(v)=8\Delta/q$ for all $v\in V$, the connected components participating in the post-shattering phase of the algorithm are of size $\poly(\Delta)\cdot \log n$, w.h.p.
\end{lemma}
\begin{proof}
	For $j\in [q/2]$ and node $v$, let $D_j(v)$ be the number of neighbors of $v$ in bucket $j$. We have $E[D_j]=2\Delta/q$ for each of the  $q/2$ buckets. By Chernoff (\cref{lem:basicchernoff}), a node $v$ has an unusually high number of neighbors ($>z(v)=8\Delta/q=(1+3)2\Delta/q$) in a given bucket w.p.\ at most $\exp(-2\Delta/q) \leq \Delta^{-12}$~,
	using $q \leq (1/6)\Delta/\ln \Delta$. A node $v$ takes part in the post-shattering phase if one of its neighbors or $v$ itself renounced its choice of bucket, i.e., if a node in its distance-$2$ neighborhood had too many neighbors in one of the buckets. This occurs w.p.\ at most $q\cdot \Delta^2 \cdot \Delta^{-12} \leq \Delta^{-9}$, and is fully determined by the random choices of nodes inside the 3-hop ball around $v$. 
	Hence, by \cref{lem:shattering}, the graph is shattered into components of size $O(\Delta^6\log n)$, w.h.p.
\end{proof}

\begin{lemma}
	\label{lem:qdividePostShatteringLLL}
	When $z(v)=8\Delta/q$, for all $v\in V$, the instances formed 
	in the post-shattering phase are LLL problems with criterion $f(d) = (q/2)\exp(-2\sqrt{d}/q)$ and $d\leq \Delta^2$. For $q\leq (1/6)\Delta/\ln \Delta$, the error probability of the LLL is upper bounded by $d^{-5}$.
\end{lemma}
\begin{proof}
	Consider the following probabilistic process. Each node picks each part in $[q]\setminus [q/2]$ u.a.r., i.e., with probability $p=2/q$. 
	For $j\in [q]\setminus [q/2]$, let $D_j$ denote the random variable describing the number of neighbors in bucket $j$. We have $E[D_j]\leq 2\Delta/q$. 
	Let $X_v$ denote the ``bad'' event that node $v$ has more than $z(v)=8\Delta/q$ neighbors in one of the $q/2$ buckets.
	We analyze the LLL formed by the events $X_v$ and their underlying variables.

	The event $X_v$ is fully determined by the random choices of direct neighbors of $v$. Hence, two bad events $X_v$ and $X_w$ are dependent on a shared variable iff $v$ and $w$ are at distance $2$ or less, and each bad event shares a variable with at most $\Delta^2$ other events.

	Therefore, the dependency graph of the LLL has maximum degree at most $d \leq \Delta^2$.
	By Chernoff (\cref{lem:basicchernoff}), $X_v$ occurs w.p.\ at most $(q/2)\exp(-2\Delta/q)$.
	Hence, the LLL has criterion $f(d) = (q/2)\exp(-2\sqrt{d}/q)$, which ranges from being polynomial to exponential depending on how small $q$ is compared to $\Delta \geq \sqrt{d}$. In the worst case, the bound $q\leq (1/6)\Delta / \ln\Delta$ implies that $f(d) \leq (\Delta/ (12\ln \Delta))\Delta^{-12} \leq d^{-5}$.
\end{proof}

\begin{proof}[Proof of \Cref{thm:qdivide}]
	The problem is solved by the algorithm above.
	The runtime is $O(1)$ rounds for the pre-shattering phase, and $\poly\log\log n$ rounds for the post-shattering phase via \Cref{lem:LOCALpostshattering}. To apply this lemma we require \Cref{lem:qdivideShattering} that shows that any component in the post-shattering phase has size $\log n \cdot \poly \Delta=\poly\log n$, w.h.p., and that \Cref{lem:qdividePostShatteringLLL}  shows that these components form LLLs with a polynomial criterion.
\end{proof}
Note that in the special case of $\Delta/q=\Omega(\log n)$ we get the stronger property that,  w.h.p., there will no post-shattering phase (see \Cref{rem:zeroRoundQdivide}). 
\section{Vertex Splitting in \LOCAL}
\label{sec:splitting}
In this section, we prove the following result on vertex splitting. 

\thmVertexSplitting*

When $\Delta$ is logarithmically larger than $k$, there is an easy solution.
\begin{restatable}{observation}{obsZeroRound}
	\label{obs:zeroRoundSplitting}
	If $k \leq \eps^2 \Delta/(9\ln n)$, the trivial zero round algorithm in which each node picks one of the $k$ parts u.a.r.\ results in a $k$-vertex splitting with discrepancy $\eps \Delta/k$, w.h.p.
\end{restatable}

\begin{proof}
	For a node $v$ and class $i$, let $D$ be the number of neighbors of $v$ that picked class $i$. Then $\Exp[D] = d_v/k$. Let $\mu:=\Delta/k$, $z=\eps\Delta/k$. By \cref{cor:handychernoff} (Chernoff bound)
	\[ \Pr[|D - \Exp[D]| \ge \epsilon \Delta/k] \le 2 \exp(-z^2/(3\mu))
	= 2 \exp(-\eps^2 \Delta /(3k))\ . \]
	This is at most $2n^{-3}$ when $k\leq \eps^2\Delta/(9\ln n)$, so by union bound over all nodes $v$ and classes $i$ we get a $k$-splitting w.h.p.
\end{proof}

\subsection{Shattering for \texorpdfstring{$\eps$}{ε}-Vertex-Splitting in \texorpdfstring{$O(1/\eps)$}{O(1/ε)} Rounds}

\label{ssec:Algorithm}
Due to \Cref{obs:zeroRoundSplitting}, the most challenging case for a  a $\poly\log\log n$-round algorithm is when $\Delta\leq \poly\log n$ and  $\Delta/k=O(\log n)$ holds. Next, we present our algorithm.

\subparagraph{\textsf{FastShattering}.} 
Find a $q$-divide $\chi$ for $q=24/\eps$. To avoid confusion between this partition of the nodes and that of the $k$-splitting we are computing, let us refer to $\chi$ as a \emph{schedule} of the nodes, made of $q$ \emph{slots}, which we denote by $N_1,\ldots,N_q$. Go through the $q$ slots of $\chi$ sequentially, and temporarily assign each node in this slot one of the $k$ parts uniformly at random. 
If a node has received too few or too many neighbors in a part when processing a slot, we retract the last batch of assignments within the neighborhood of that node and freeze those nodes. We also freeze all nodes within distance $3$ that are in later slots. All non-frozen nodes (in slot $j$) keep their assignment permanently. The frozen nodes then get solved in post-shattering (along with all neighbors acting as events, including non-frozen neighbors). For each $j\in [q]$, there is one such post-shattering instance stemming from nodes that were frozen when processing slot $j$.  

\begin{figure}[ht]
	\centering
	\includegraphics[width=0.6\textwidth]{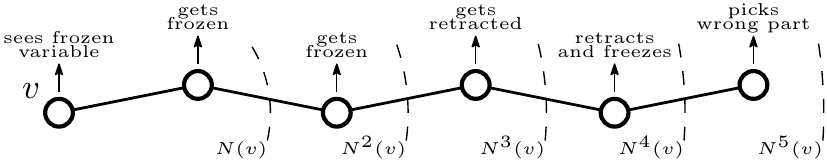}
	\caption{Whether a node joins the post-shattering instance depends on random choices at distance up to $5$.}
	\label{fig:lllshattering}
\end{figure}

For the rest of this section, set the number of slots to $q=24/\eps$ and define the following threshold parameters for the pre-shattering and post-shattering phase $z^{\pres}_j(v)=z^{\post}_j(v) = \eps^2\Delta/(72k)$ for all $j\in[q], v\in V$.

\subparagraph{Detailed description of \textsf{FastShattering}.}
During the course of the algorithm nodes are either \frozen or non-\frozen. Initially, all nodes are non-\frozen. 
\textsf{Pre-shattering:} 
After computing the $q$-divide, we iterate through the slots $1, \ldots, q$. In each iteration, we temporarily assign the non-\frozen nodes in slot $j$ by sampling each u.a.r.\ into one of the $k$ parts. Next, we formalize the event that retracts these assignments. Fix a node $v\in V$, a slot $j\in[q]$ and a part $i\in [k]$. Let $N_j(v)$ be the neighbors of node $v$ in slot $j$. Let $\hat{d}_j(v)$ denote the \emph{live degree} of node $v$ when processing slot $j$, i.e., the number of vertices in $N_j(v)$ that are not \frozen just before processing slot $j$. Let $D_{i,j}$ denote the number of neighbors of $v$ in $N_j(v)$ that are temporarily assigned to part $i$. Event $\mathcal{B}_{i,j}^{\pres}(v)$ holds if $D_{i,j}$ deviates from its expectation $\Exp[D_{i,j}]=\hat{d}_j(v)/k$ by more than the threshold parameter $z^{\pres}_j(v)$. Let $\mathcal{B}_{j}^{\pres}(v)=\bigvee_{i\in k}\mathcal{B}_{i,j}^{\pres}(v)$ be the event that $v$ sees such a large deviation from its expectation in some part $i$. Suppose after sampling event $\mathcal{B}_{j}^{\pres}(v)$ occurs, then node $v$ undoes the temporal assignment of all neighbors in slot $j$, i.e., of all nodes in $N_j(v)$, and additionally freezes all unassigned nodes in distance $3$, i.e., the nodes in $\{u \in \cup_{j' > j} N_j(v) : d(u,v) \le 3 \}$. Add all nodes that become frozen when processing slot $j$ to $\Bad_j$. All (temporal) assignments that do not undergo a retraction are kept permanently. 
While \frozen nodes do not sample colors, each node monitors how its neighbors are being colored and thus yields an \emph{event node} in each of the $q$ iterations, regardless of whether it is \frozen or not.

\textsf{Post-shattering:}  For each $j\in [q]$ there is a separate post-shattering LLL instance containing a variable for each node in $\Bad_j$ and a bad event node  $\mathcal{B}^v_{j}$ for each node $v$ with a neighbor in $\Bad_j$.  The random process of the $j$-th LLL is as follows: Each node in $\Bad_j$ picks one of the $k$ parts independently and u.a.r. For a node $v$ the number of neighbors in $\Bad_j$ is denoted by $f_j(v)$. Let $F_{i,j}(v)$ be the number of neighbors of $v$ in part $i$ (restricted to neighbors in $\Bad_j$). Event $\mathcal{B}_{i,j}^{\post}(v)$ holds if $F_{i,j}$ deviates from its expectation $\Exp[F_{i,j}(v)]=f_j(v)/k$ by more than the threshold parameter $z^{\post}_j(v)$. The \emph{bad event} $\mathcal{B}_{j}^{\post}(v)=\bigvee_{i\in k}\mathcal{B}_{i,j}^{\post}(v)$ holds if $v$ sees such a large deviation from its expectation in some part $i$. 
In \Cref{lem:postShatteringLLL,lem:shattering} we show that for each $j\in[q]$ we indeed obtain an LLL with polynomial criterion that can be solved via  \Cref{lem:LOCALpostshattering} in the \LOCAL model. All $q$ instances are solved in parallel; their deviations add up to $(\eps/3)\cdot \Delta/k$ as shown in \Cref{lem:discrepancy}.

\textbf{Intuition for the runtime:} The \textsf{pre-shattering} phase runs in $O(q)=O(1/\eps)$ rounds. The \textsf{post-shattering} phase runs in $\poly\log\log n$ rounds for the following reason. Each component in each of the $q$ \textsf{post-shattering} instances forms an LLL and is of size $N=\poly(\Delta)\cdot \log n=\poly\log n$, see \Cref{lem:fastShatteringShatters}. As all components are independent, they can can be solved in parallel in $\poly\log N=\poly\log\log n$ rounds in the \LOCAL model via \Cref{lem:LOCALpostshattering}.

\subparagraph{Notation.}
We summarize and extend the notation that we need for the analysis. 
\begin{itemize}
	\item $V_1, \ldots, V_k$ parts (changing throughout the algorithm),
	\item $N(v)$ neighbors of $v$ in $G$, $N_j(v)$ neighbors of $v$ in slot $j$, $\hat{N}_j(v)\subseteq N_j(v)$ are the live neighbors of $v$ in bucket $j$, i.e., the unfrozen neighbors of $v$ in slot $j$ just before slot $j$ is processed, $F_j(v)\subseteq N(v)$ neighbors of $v$ in the $j$-th post-shattering instance,
	\item $d(v)=\card{N(v)}$, $d_j(v) = \card{N_j(v)}$, $\hat{d}_j(v) = \card{\hat{N}_j(v)}$, $f_j(v) = \card{\hat{F}_j(v)}$~.
\end{itemize}

\begin{observation}
	\label{obs:oneRetraction}
	In \textsf{FastShattering}, any node can have at most one slot in which (some) of its neighbors get their part assignment undone. 
\end{observation}
\begin{proof}
	Let $v$ be a node with a neighbor $u\in N_j(v)$ that has its  assignment retracted during slot $j$. Then $u$ is adjacent to a node that detected that too few or too many of its neighbors were assigned a given part  when processing slot $j$. That node is at distance at most $2$ from $v$, and it freezes the nodes in slots higher than $j$ within distance $3$. Therefore, all the unassigned neighbors of $v$ are frozen, and $v$ will not see another retraction in its neighborhood (in fact, it will not even see an assignment).
\end{proof}

\subsection{Vertex Splitting: Bounding the Discrepancy}
\label{s:discAnalysis}
In this section, we bound the deviation in the number of neighbors that a node $v$ sees in the $i$-th part from $d(v)/k$.
For a node $v$ let $N_{\pres}(v)$ ($N_{\post}(v)$) be the neighbors of $v$ that are permanently assigned to a part in the pre-shattering (post-shattering) phase. Also,  let $d_{\pres}(v)=|N_{\pres}(v)|$ and $d_{\post}(v)=|N_{\post}(v)|$.
Recall, the definition of $z^{\pres}_j(v) = z^{\post}_j(v) = \eps^2/(72k)$ and $q=24/\eps$, which immediately yields the following claim. 
\begin{claim}
	\label{claim:boundZj}
	We have 
	$\sum_{j\in [q]}z^{\pres}_j(v)=\sum_{j\in [q]}z^{\post}_j(v)\leq \eps/3 \cdot \Delta/k$~.
\end{claim}

\medskip

\begin{lemma}
\label{lem:discrepancy}
	In the final assignment $V_1, \ldots, V_k$, i.e., after the pre-shattering and post-shattering phase, we have the following guarantees on the split for each part $i\in[k]$:
	\begin{enumerate}
		\item Node $v$ has  $d_{\pres}(v)/k \pm 2\eps/3 \cdot \Delta/k$ neighbors in $V_i\cap N_{\pres}(v)$. 
		\item  Node $v$ has  $d_{\post}(v)/k \pm \eps/3 \cdot \Delta/k$ neighbors in $V_i\cap N_{\post}(v)$. 
	\end{enumerate}
	In total, for each $i\in [k]$ any node $v$ has $d(v)/k \pm \eps\Delta/k$ neighbors in $V_i$.
\end{lemma}
\begin{proof}
	We first prove the first claim. The discrepancy (deviation from expectation) for a node $v$ comes from two sources: (a) slots with neighbors that got retracted; (b) other slots. We bound both separately.  Consider a vertex $v$ and fix a part $V_i, i\in [k]$. For the rest of the proof let $z_j=z^{\pres}_j(v)$. We partition the vertices in $V_i\cap N_{\pres}(v)$ according to the $q$ slots $N_1(v),\ldots N_q(v)$.  Due to \Cref{obs:oneRetraction} for at most one $j$ does $N_j(v)$ contain nodes whose values were retracted. Denote this $j$ (if any) by $j^*$, otherwise set $j^*=\bot$.
	\begin{claim} 
		\label{claim:discrepancyBucketRetracted}
		If $j^*\neq \bot$, then  
		$|V_i\cap N_{j^*}(v) \cap N_{\pres}(v)|\leq \hat{d}_{j^*}(v)/k +z_j$.
	\end{claim} 
	\begin{proof}
		If $v$ caused the retraction then, $N_{j^*}(v)\cap N_{\pres}(v)=\emptyset$, as $v$ retracted all assignments of nodes in $N_{j^*}(v)$ and froze the nodes (they will only be assigned in the post-shattering phase). Now consider the case that $v$ did not cause the retraction, i.e., $\mathcal{B}^{\pres}_j$ does not occur, and let $X_i$ be the nodes in part $i$ in the temporal assignment of nodes in slot $j$ before any retractions happened (also before the ones caused by nodes $u\neq v$). Since $\mathcal{B}^{\pres}_j$ does not occur, we have $|X_i\cap N_{j^*}|\in \hat{d}_{j^*}/k\pm z_j$. Some nodes of $X_i$ might get retracted by other nodes $u\neq v$, but we obtain $|V_i\cap N_{j^*}(v) \cap N_{\pres}(v)|\leq |X_i\cap N_{j^*}(v)| \leq \hat{d}_{j^*}(v)/k+z_j$.
	\end{proof}
	\begin{claim}
		\label{claim:discrepancyBucket}
		For each $j\notin [q]\setminus j^*$, we obtain $|V_i\cap N_{\pres}(v)\cap N_j(v)|=\hat{d}_j(v)/k\pm z_j$.
	\end{claim}
	\begin{proof}
		Since $j\neq j^*$ there are no retracted variables in $N_j(v)$. If the bound in the claim does not hold, then $\mathcal{B}^{\pres}_{j}(v)$ would have occurred after the sampling, and $v$ would have retracted, a contradiction. 
	\end{proof}
	\begin{claim}
		\label{claim:sumDegrees}
		$\sum_{j\in [q],j\neq j^*} \hat{d}_j(v)
		\leq d^{\pres}(v)
		\leq \sum_{j\in [q]} \hat{d}_j(v)$.
	\end{claim}
	In the following we omit the explicit dependence on $v$, e.g., we write $\hat{d}_{j}$ instead of $\hat{d}_{j}(v)$. 
	Using, $\sum_{j\in [q]}z_j(v)\leq \eps \Delta/(3k)$ (\Cref{claim:boundZj}),
	$\hat{d}_{j*}\leq d_{j*}\leq 8\Delta/q = \eps \Delta/3$ (from the properties of a $q$-divide), bounds on $|V_i\cap N_{\pres}(v)\cap N_j(v)|$ (\Cref{claim:discrepancyBucketRetracted,claim:discrepancyBucket}), and \Cref{claim:sumDegrees}
	we obtain
	\begin{align*}
		|V_i \cap N_{\pres}(v)|
		& \leq \sum_{j\in [q]}\parens*{\hat{d}_j/k+z_j} 
		\leq \parens{d^{\pres}+\hat{d}_{j^*}}/k+\sum_{j\in[q]} z_j
		\leq d^{\pres}/k + 2\eps \Delta/(3k) ~~~\text{and}
		\\
		|V_i \cap N_{\pres}(v)|
		& \geq \sum_{j\in [q], j\neq j^*}\parens*{\hat{d}_j/k-z_j}
		\geq \parens*{d^{\pres}-\hat{d}_{j^*}}/k- \eps \Delta/(3k)
		\geq d^{\pres}/k-2\eps \Delta/(3k)~.
	\end{align*}
	
	For the second part of the claim, fix again some $i\in [k]$ and a node $v$. There are $q$ separate post-shattering instances. Recall,  the set of neighbors of a node participating in the $j$-th instance  is denoted by $F_{j}(v)$ and  $f_j(v)=|F_{j}(v)|$. The solution to the LLL instance yields 
	\begin{align}
		|V_i\cap F_j(v)\cap N_{\post}(v)|=f_j(v)/k\pm z^{\post}_j(v).
	\end{align}
	Summing over all $q$ post-shattering instances, using $d_{\post}(v)=\sum_{j\in [q]}f_j(v)$ and using \Cref{claim:boundZj} to bound $\sum_{j\in [q]}z^{\post}_j(v)\leq \eps\Delta/3\leq \eps\Delta/2$ yields the second part of the claim. 
\end{proof}

\subsection{Analysis of Bad Event Probabilities}

Throughout our analysis of the pre-shattering and post-shattering parts of our algorithm, we consider random processes and events which are essentially always the same: nodes in some subgraph each pick a random bucket u.a.r.\ independently from other nodes, and for each node we analyze the probability that the number of neighbors that pick a given bucket deviates too much from expectation.
Recall, that we set $q=24/\eps$ and $z=\eps^2\Delta/(72k)$ earlier.

\begin{claim}
	\label{claim:deviationProb}
	Let $k$, $N$ be positive integers. Let $D$ be a sum of at most $N$ independent Bernouilli random variables of parameter $1/k$, and let $z \leq N/k$. Consider the event $\mathcal{B}$ that $D$ deviates from its expectation by more than $z$. $\Pr(\mathcal{B}) \leq 2e^{-z^2k/(3N)}$.
	
	In particular, for $N=\Delta$, $k\leq \eps^4\Delta/(2^{19}\ln \Delta)$ and $z=\eps^2\Delta/(72k)$ we obtain $\Pr(\mathcal{B})\leq  \Delta^{-24}$. If $D$ is a sum of only $N=8\Delta/q$ variables, $k\leq \eps^3\Delta/(2^{17}\ln \Delta)$ suffices for the same bound. 
\end{claim}
Throughout this paper, $D$ is taken to be a sum of indicator random variables associated to a set of nodes. More precisely, for a subset of nodes in a neighborhood $N(v)$, we consider the sum of the random variable indicating whether each node chose a specific part $i$ out of $k$ choices. 
\begin{proof}
	The general bound on $\Pr(\mathcal{B})$ is from \cref{cor:handychernoff} (Chernoff bound).
	
	When $N=\Delta$, $k \leq \eps^4 \Delta /(2^{19} \ln \Delta)$ and $z = \eps^2 \Delta /(72k)$, $\exp(-z^2/(3N))$ simplifies to $\exp(-\eps^4 \Delta/(3\cdot 72^2 k)) \leq \Delta^{-24}$.
	When $N=8\Delta/q$ (recall $q=24/\eps)$, $k \leq \eps^3 \Delta /(2^{17} \ln \Delta)$ and the same $z$ as before, $\exp(-z^2k/(3N))$ simplifies to $\exp(-\eps^3 \Delta/(72^2 k)) \leq \Delta^{-24}$.
\end{proof}

\subsection{Analysis of \textsf{FastShattering}}
\label{ssec:fastShatteringAnalysis}
Next, we show that the post-shattering instances consist of small connected components. 
\begin{restatable}{lemma}{lemFastShatteringShatters}
	\label{lem:fastShatteringShatters}
	After  \textsf{FastShattering}, each connected component in each of the $q$ post-shattering instances is of size $\Delta^{10}\log n$, w.h.p.
\end{restatable}
\begin{proof}[Proof of \cref{lem:fastShatteringShatters}]
	Let us focus on one post-shattering instance, instance number $j$, formed of both the nodes in $\Bad_j$ that were frozen while processing slot $j$ and all their incident 'events nodes'. Let us say that a node $v$ \emph{triggers} if one of the events $\mathcal{B}_{i,j}^{\pres}$, $i\in [k]$ occurs, i.e., if $D_{i,j}(v)$ deviates too much from expectation. That a node triggers is entirely determined by the random choices of its neighbors.  By \Cref{claim:deviationProb} (applied with $N=8\Delta/q$, $z=z_j^{\pres}(v)=\eps^2\Delta/(72k)$, and $D_{i,j}(v)$), the probability that a node triggers is at most $\Delta^{-24}$. 
	A variable is \frozen if it is within distance $3$ of a triggering node.
	A node $v$ joins the post-shattering instance if it is \frozen itself or one of its neighbors is \frozen, which depends on whether nodes within distance $4$ of $v$ trigger or not, which itself is entirely determined by the random choices within distance $5$ of $v$. Thus, whether two nodes at distance $10$ participate in the $j$-th post-shattering instance depends on two sets of non-overlapping random variables from the processing of slot $j$. 
	
	By a union bound over the $\Delta^4$ nodes in the $4$-hop neighborhood and the $k$ parts, a node participates in the $j$-th post-shattering instance w.p.\ at most $k(\Delta^4) \Delta^{-24} \leq \Delta^{-19}$. By \cref{lem:shattering}, the resulting connected components of the post-shattering instance are all of size $O(\Delta^{10} \log_\Delta n)$, w.h.p.
\end{proof}

\subsection{\textsf{Post-shattering}}
In this section we show that the $q$ post-shattering instances are indeed LLLs.
\begin{lemma}
	\label{lem:postShatteringLLL}
	Each connected component in each of the $q=O(1/\eps)$ post-shattering instances forms an LLL with dependency degree $d'$, bad events' probabilities upper bounded by $p'$ such that the polynomial criterion $d'^8p'<1$ holds. In \LOCAL, the dependency graph can be simulated with $O(1)$ overhead in the communication network $G$. 
\end{lemma}
\begin{proof}
	Consider a post-shattering instance $j \in [q]$. The LLL is formally defined in \Cref{ssec:Algorithm}. Recall, that in the associated random process each node in $\Bad_j$ joins one of the  $k$ parts u.a.r.\ and that there is a bad event $\mathcal{B}^{\post}(v)=\bigvee_{i\in [k]}\mathcal{B}_i^{\post}(v)$ for each node with neighbor in $\Bad_j$. The  event $\mathcal{B}^{\post}(v)$ occurs if too many or too few neighbors join the $i$-th part. Thus, a bad event only depends on the randomness of adjacent nodes and the dependency degree is at most $d'=\Delta^2$. 
	
	By \cref{claim:deviationProb} (applied with $N=\Delta$, and $z=z_j^{\post}(v)=\eps^2\Delta/(72k)$), the probability that $\mathcal{B}_i^{\post}(v)$ holds is at most $\Delta^{-24}$ if $k\leq \eps^4\Delta/(2^{19}\log \Delta)$. With a union bound over all $k$ parts we obtain the upper bound $p'=k\Delta^{-24}=\Delta^{-23}$ for the probability of each bad event. 
	
	Hence, we obtain $p'd'^{11}<1$.
\end{proof}

\subsection{Proof of Theorem~\ref{thm:vertexSplitting}}

Assume $k\leq  \eps^{4}\Delta/(2^{19}\log \Delta)$ and recall that $\Delta\leq \poly\log n$. The runtime of \textsf{Fast-Shattering} is linear in the number of slots, i.e., $O(q)=O(1/\eps)$. 
Next, we show that the post-shattering instances meet the requirements of \Cref{lem:LOCALpostshattering}. Due to \Cref{lem:fastShatteringShatters} each connected component is of size $\poly(\Delta)\log n=\poly\log n$, w.h.p. Further, due to \Cref{lem:postShatteringLLL} each such component forms an LLL with polynomial criterion and the dependency graph can be simulated with $O(1)$ overhead in the communication network $G$. Thus, we can apply \Cref{lem:LOCALpostshattering} (in parallel for all $q$ instances) and obtain a runtime of $\poly\log\log n$ rounds for the post-shattering phase. 
\Cref{lem:discrepancy} shows that the deviation of $|N(v)\cap V_i|$ from $d(v)/k$ is upper bounded by  $\eps \Delta/k$ for any node $v\in V$.

\section{Vertex Splitting in \CONGEST}
\label{sec:splittingCongest}
We obtain the following theorem for vertex splitting and bipartite vertex splitting. 

\begin{restatable}{theorem}{thmCONGESTSplitting}
	\label{thm:CONGESTsplittingBoth}
	There exists a universal constant $\inccte>0$ s.t.:
	For 
	$\eps>0$,  $\Delta\leq \poly\log n$, and 
	$k\leq  \cte \cdot(\eps^{4} \Delta/(\ln \Delta\log^2\log n))$,
	there are distributed \CONGEST algorithms to solve the $k$-vertex splitting problem with discrepancy $\eps\Delta/k$  and to solve the bipartite $k$-vertex splitting problem with discrepancy $\eps\DeltaL/k$ in  $O(1/\eps)+\poly \log\log n$ rounds. 
\end{restatable} 

\Cref{thm:CONGESTsplittingBoth} requires $k$ to be a $O(\log^2\log n)$ factor smaller than in  \Cref{thm:vertexSplitting,thm:bipartiteSplitting}.

As the pre-shattering phase of \Cref{thm:vertexSplitting,thm:bipartiteSplitting} immediately works in the \CONGEST model, the main challenge to prove \Cref{thm:CONGESTsplittingBoth} is to design a new post-shattering method. 
Recall, the post-shattering phase in the \LOCAL model, i.e, the core steps of \Cref{lem:LOCALpostshattering}. Each connected component in the post-shattering phase forms an LLL with a polynomial criterion and has $N=\poly(\Delta)\cdot \log n=(\poly\log\log n)\cdot\log n$ nodes. This small size allows to compute a network decomposition (see \Cref{ssec:networkDecomp}) with $\poly\log\log n$ cluster diameter and $O(\log\log n )$ color classes with distance $s=\Omega(\log N)=\Omega(\log \log n)$ between clusters of the same color.  The latter is sufficient to derandomize the $O(\log N)=O(\log\log n)\ll s$ round LLL algorithm from \cite{CPS14}. The details of the derandomization are not important, but it is based on \emph{gathering all information in the cluster and close-by nodes}. In the \LOCAL model, this can be done in time that is linear in the cluster diameter, i.e., in $\poly\log\log n$ rounds. One can show that in the \CONGEST model all information of a cluster can be encoded with $N\cdot \poly\log \log n$ bits. By using a pipelining argument (\Cref{cor:treeAggregationBetter})  and that the \textsf{bandwidth} of the \CONGEST model is $\Theta(\log n)$ bits, one can aggregate all of this information at a cluster leader in  $N\cdot \poly\log\log n/\textsf{bandwidth} +\textsf{clusterdiameter}$ rounds, as done in \cite{MU21}. 
For $\Delta=\poly\log\log n$, we obtain $N=\log n\cdot \poly\log\log n$ and this method runs in $\poly\log\log n$ rounds. 
In summary, we obtain the following theorem\footnote{The proof of  \Cref{lem:CONGESTpostshatteringSmallDegree} appears in \Cref{app:postshattering}. It is similar to an \CONGEST LLL algorithm in \cite{MU21}  for instances of size $N=O(\log n)$ and the case of $d=O(1)$. In fact, following all dependencies on $d$ (and a slightly increased $N$) in the proof of \cite{MU21} yields an algorithm with runtime $\poly (d, \log\log n)$, which yields the desired runtime whenever $d=\poly\log \log n$. }  and the corollary thereafter.

\begin{restatable}[\cite{MU21}]{lemma}{CONGESTpostshatteringSmallDegree}
	\label{lem:CONGESTpostshatteringSmallDegree}
	There is a randomized \CONGEST algorithm with $\textsf{bandwidth}=\Theta(\log n)$ 
	for LLL instances of size $N\leq \log n\cdot \poly\log\log n$, dependency degree $d\leq \poly\log\log n$ and error probability $p<d^{-4}$, that runs in $\poly \log \log n$ rounds.
	
	The algorithm works with an ID space that is exponential in $N$ and is correct w.h.p in $n$. 
\end{restatable}
\begin{corollary}
	\label{cor:LLLCONNGESTsmalldegree}
	There is a randomized \CONGEST algorithms for LLL with error probability $p$, dependency degree $d$ and criterion $p<d^{-8}$ that uses $\poly\log\log n$ rounds, whenever $d\leq \poly\log\log n$. Here, the dependency graph is also the communication network. 
\end{corollary}
\begin{proof}
	The shattering framework of \cite{FG17}, w.h.p.,  reduces to the LLL problem to LLL problems with error probability $p'$, the same dependency degree $d$ and criterion $p'<d^{-4}$ on instances of size $N=\log n\cdot \poly d$. These can be solved in $\poly\log\log n$ rounds  via  \Cref{lem:CONGESTpostshatteringSmallDegree}.
\end{proof}
For $\Delta\gg \poly\log\log n$, any such \emph{gather all information} approach inherently requires significantly larger runtimes.  
The main ingredient for \Cref{thm:CONGESTsplittingBoth} is a new method for solving  the vertex splitting instances in the post-shattering phase,  that can deal with degrees as large as $\Delta=\poly\log\log n$ while using only $\poly\log \log n$ rounds. 
We prove the following theorem.

\begin{restatable}{lemma}{CONGESTpostshattering}
	\label{lem:CONGESTpostshattering}
	There exists a universal constant $\inccte>0$ s.t.:
	For any $\eps>0$ and any 
	$k \leq \cte \cdot (\eps^2\Delta/(\log\Delta\log^2\log n))$,
	there is a $\poly\log\log n$-round randomized \CONGEST algorithm with $bandwidth=\Theta(\log n)$ 
	that computes a $k$-vertex splitting with discrepancy $\eps\Delta/k$ on instances of size $N\leq \poly\log n$.

	The algorithm works with an ID space that is exponential in $N$ and is correct w.h.p in $n$. 
\end{restatable}

The proof of \Cref{lem:CONGESTpostshattering} uses  network decompositions that we introduce in \Cref{ssec:networkDecomp}, before proving the lemma in \Cref{ssec:CONGESTpostshattering}. In \Cref{ssec:proofthmCONGESTsplitting}, we prove \Cref{thm:CONGESTsplittingBoth}.

\subsection{Network Decomposition}
\label{ssec:networkDecomp}
A \emph{weak distance-$s$ $(C,\beta)$-network decomposition with congestion $\kappa$} is a partition of the vertex set of a graph into clusters $\cC_1,\ldots,\cC_p$ of (weak) diameter $\leq \beta$ , together with a color from $[C]$ assigned to each cluster such that clusters with the same color are further than $s$ hops apart. Additionally,  each cluster has a communication backbone, a Steiner tree of radius $\leq \beta$, and each edge of $G$ is used in at most $\kappa$ backbones. For additional information on such decompositions we refer the reader to \cite{MU21,GGR20}. For the sake of our proofs we only require that such decompositions can be computed efficiently (\Cref{thm:networkDecomp}) and that one can efficiently aggregate information in all clusters of the same color in parallel in time that is essentially proportional to the diameter $\beta$ (see \Cref{cor:treeAggregationBetter} in \cref{app:networkDecompositionRouting} for the precise statement). 

\begin{theorem}[\cite{MU21}]
	\label{thm:networkDecomp}
	For any constant $C>0$ and $s\in \poly\log\log n$, there is a deterministic \CONGEST algorithm with bandwidth $b$ that, given a graph $G$ with at most $n$ nodes and unique $b$-bit IDs from an exponential ID space, computes  a weak $(C\log n, O(s/C\cdot \log^3 n))$-network decomposition with cluster distance $s$ and congestion $O(s\cdot \log^2 n)$ in $O(\log^7 n\cdot s^2)$ rounds.
\end{theorem}

\subsection{Efficient Post-shattering in \CONGEST (Proof of \texorpdfstring{\Cref{lem:CONGESTpostshattering}}{Lemma~\ref{lem:CONGESTpostshattering}})}
\label{ssec:CONGESTpostshattering}

In order to devise an efficient \CONGEST post-shattering algorithm, we decompose each small component into small  clusters via the network decomposition algorithm from \Cref{thm:networkDecomp}. Then, the objective is to iterate through the color classes of the decomposition and when processing a cluster we want to assign all nodes in that cluster to a part. When doing so we ensure that each node of the graph obtains a discrepancy of at most $(\eps/Q)\Delta/k$ in each iteration. Hence, over the $Q$ iterations, each node's discrepancy adds up to at most $\eps\Delta/k$.

\begin{proof}[Proof of \Cref{lem:CONGESTpostshattering}]
	
	First, compute a distance-$3$ network decomposition of the graph with $Q=2\log\log n$ colors via \Cref{thm:networkDecomp}. Then, iterate through the color classes of the network decomposition, processing all clusters of a color class as it gets considered.

	When processing a cluster $\cC$, we set up a new instance of the vertex splitting problem as follows: Let $V^{\Lmark,\cC}=N(\cC)$ be all nodes that have a neighbor in $\cC$; $V^{\Lmark,\cC}$ may contain many nodes of $\cC$ itself. Each node of $\cC$ is supposed to join one one the parts $V^{\cC}_1,\ldots,V^{\cC}_k$ such that for each $i\in [k]$ each node in $v\in V^{\Lmark,\cC}$ has $d_{\cC}(v)/k\pm \eps/Q \cdot \Delta/k$ neighbors in $V^{\cC}_i$. After processing all clusters we set $V_i=\bigcup_{\text{cluster}\  \cC}V^{\cC}_i$. As clusters processed at the same time are in the same color class, they have distance-$3$, and no node has neighbors in more than one simultaneously processed cluster. Hence, the deviation of the number of neighbors into one $V_i$ from $d(v)/k$ is bounded by $Q\cdot \eps/Q \cdot \Delta/k=\eps\Delta/k$. 
	
	The bounds on $k$ and $Q$ imply that the problem that we solve when processing one cluster is an LLL $\cL_{\cC}$ with a polynomial criterion: Variables and the random process are given by the nodes of $\cC$ choosing one of the parts $V^{\cC}_1, \ldots, V^{\cC}_k$ uniformly at random. For a node $v\in V^{\Lmark,\cC}$ introduce a bad event $\cB^{\cC}_v$ that holds if for any $i\in [k]$ node $v$ does not have $d_{\cC}(v)/k\pm \eps/Q\cdot \Delta/k$ neighbors in part $V^{\cC}_i$. Due to the distance between clusters no node can have a bad event for more than one of the simultaneously processed clusters. 
	
	Due to  \Cref{claim:deviationProb} (applied with $N=\Delta$, $z=\eps/Q \cdot \Delta/k$), we obtain that $\Pr(\cB^{\cC}_v)\leq k\cdot  \exp(-\eps^2\Delta / (3Q^2k))$. Plugging in $Q = 2 \log \log n$ and $k \leq C \eps^2 \Delta / (\log \Delta \log^2 \log n)$, we get $\Pr(\cB^{\cC}_v)\leq k\cdot  \exp(-\log \Delta / (12 C)) \leq \Delta^{-19}$ for $C \leq 2^{-8}$.
	As the dependency degree is at most $\Delta^2$, we obtain an LLL with a polynomial criterion of exponent  $9$.

	The goal is to assign all nodes of $\cC$ to a part such that   all bad events $\cB^{\cC}(v)$ for $v\in V^{\Lmark,\cC}$ are  avoided. In order to do so,  we run $\ell=6\log n$ parallel instances of the LLL algorithm of \cite{CPS14} on $\cL_{\cC}$, each running for $O(\log N)=O(\log \log n)$ rounds. At the end of the proof we reason that these $\ell$ instances can indeed be run efficiently in parallel, for now, we continue with the remaining steps of the algorithm. We say that an instance is \emph{correct} for an event of $\cL_{\cC}$ if it is avoided under the computed assignment of the instance. By the properties of the algorithm of \cite{CPS14}, each instance is correct for all events of $\cL_{\cC}$ with probability $\geq (1-1/N)\geq 1/2$. Hence, with probability $1-1/2^{\ell}=1-1/n^6$ one of the $\ell$  instances is \emph{correct} for all events of $\cL_{\cC}$.  Then, each node holding an event of $\cL_{\cC}$ determines which instances are correct, and the nodes agree on a \emph{winning} instance, i.e., one that is correct for all of them. 
	
	Assume that nodes know in which instance their bad events are avoided. Then, agreeing on a winning instance can be done efficiently as follows: Let each such node hold a bit string of length $\ell=O(\log n)$ in which the $j$-th bit indicates whether the bad event is avoided in the outcome of the $j$-th instance. All nodes can agree on a winning instance in time linear in the cluster's weak diameter by computing a bitwise-\textsf{AND} of the bitstrings via \Cref{cor:treeAggregationBetter}.
	
	In order to determine the status of its events in each of the $\ell$ instances, node $v$ only needs to know which part each neighbor has chosen in which instance. As there are only $k$ parts, the index of the part can be communicated with $O(\log k)$ bits. Hence, a node $u$ can inform each neighbor about the parts node $u$ chose  in all $\ell$ instances by communicating $\ell \cdot O(\log k)=O(\log n\log\log n)$ bits over each incident edge. Using  $\textsf{bandwidth}=\Theta(\log n)$, this requires $O(\log \log n)$ rounds. The same reasoning is also sufficient to run the $\ell$ instances of \cite{CPS14} in parallel.  In one iteration of \cite{CPS14}, the variables of local ID minima in the graph induced by violated events are re-sampled.  We just reasoned that a node can determine the status of its events in each of the $\ell$ instances in $O(\log\log n)$ rounds, and with an additional round we can compute a set of local ID minima of violated events for each instance. Then, nodes can inform neighbors about the instances in which they need to re-sample their part.
\end{proof}

\subsection{Proof of Theorem \ref{thm:CONGESTsplittingBoth}}
\label{ssec:proofthmCONGESTsplitting}
The pre-shattering phase of computing the $q$-divide can immediately be implemented in the \CONGEST model. Its post-shattering phase is replaced with the stronger $q$-vertex splitting result of \Cref{lem:CONGESTpostshattering} (with $\eps=1$ and $k=q=24/\eps$) that runs in $\poly\log\log n$ rounds. Note that
the hypotheses of \cref{thm:CONGESTsplittingBoth} assume that $\Delta/(\log \Delta \log^2 \log n)$ is greater than an absolute constant $1/c_4$. With $c_4$ chosen s.t.\ $c_4 \leq c_5/24$, $q$ satisfies the hypotheses of \cref{lem:CONGESTpostshattering}.

The pre-shattering of the main algorithm can also immediately be implemented in the \CONGEST model. For each of its $q$ post-shattering instances we use \Cref{lem:CONGESTpostshattering} with $\eps^2/72$ and the same $k$. Using the proof of \Cref{lem:discrepancy},  the total discrepancy of the pre-shattering and the post-shattering phase is upper bounded by $(2\eps/3) \Delta/k$ and $(\eps/3) \Delta/k$, respectively.

\section{Application: \texorpdfstring{$(1+\eps)\Delta$}{(1+epsilon)∆}-edge coloring}
\label{sec:edgeColoring}
In this section, we first prove the \LOCAL version of the following theorem. 
\thmEdgeColoring*

We use the following result based on prior work to color small degree graphs. 

\begin{theorem}[\cite{FG17,EPS15,CHLPU19}]
	\label{thm:edgeColoringSlow}
	For any constant $\eps>0$, there is an absolute constant $\Delta_0$ such that for $\Delta \geq \Delta_0$, there is a randomized \LOCAL algorithm with runtime $O(d^2)+\poly\log\log n$ for $(1+\eps)\Delta$-edge coloring where $d=\poly\Delta$. 
\end{theorem}
The papers \cite{EPS15,CHLPU19} both solve the $(1+\eps)\Delta$-edge coloring problem via a constant number of LLL iterations (for constant $\eps>0$). Their dependency graph can be simulated in the original network with $O(1)$ overhead and has dependency degree $d=\poly\Delta$. Plugging in the runtime of $O(d^2)+\poly\log\log n$ for solving such LLLs by \cite{FG17} yields \Cref{thm:edgeColoringSlow}.

\subparagraph{High level overview $(1+\eps)\Delta$-edge coloring algorithm.} We recursively (two recursion levels)  partition the edge set of $G$ into parts that induce small degree subgraphs. Then, we color each subgraph with a disjoint color palette. More detailed, first we partition the edge set into $k=\Theta(\eps^2\Delta/\log n)$ parts such that each part induces a graph of maximum degree at most $\Delta'=\poly\log n$. Then, in another recursive step we partition the edge set of each of these parts further into $k'=\Theta(\eps^4\Delta'/\log^2\log n)$ parts, each with maximum degree $\Delta''=\poly\log\log n$. We obtain $k\cdot k'$ subgraphs, each with maximum degree at most $\Delta''$. We color each part with a disjoint color palette with $(1+\eps/10)\Delta''$ colors via \Cref{thm:edgeColoringSlow} in $O((\Delta'')^2)+\poly\log \log n=\poly\log\log n$ rounds. The colors of the $k\cdot k'$ subgraphs sum up to $(1+\eps)\Delta$ colors in total. 
\begin{proof}[Proof of \Cref{thm:edgeColoring}, {\LOCAL}]
	If $\Delta\leq \poly\log\log n$ we skip the first two steps of the algorithm and immediately apply \Cref{thm:edgeColoringSlow} to compute a $(1+\eps)\Delta$-edge coloring in $\poly\log\log n$ rounds.  If  $\Delta\leq \poly\log n$, we skip the first step and set  $\Delta'=\Delta$, $k=1$ and $G_1=G$, otherwise we first partition the graph into $k=(\eps/6)^2\Delta/(9\log n)$ subgraphs $G_1,\ldots, G_k$, each with maximum degree $\Delta'=\Delta/k+(\eps/6) \cdot \Delta/k=\poly\log n$. To this end, let each edge uniformly at random and independently join one of the $G_i$'s. The same Chernoff bound as in \Cref{obs:zeroRoundSplitting} shows that w.h.p., the maximum degree of each $G_i$ is upper bounded by $\Delta'$.

	In the next step, we set $k'=c_4(\eps')^4\Delta'/\log^2\log n$, $\eps'=\eps/6$, and use \cref{thm:CONGESTsplittingBoth} to  split each $G_i, i\in [k]$ in parallel into $k'$ graphs $G_{i,j}, j\in [k']$, each of maximum degree $\Delta''=\Delta'/k'+\eps'\Delta'/k'=\poly\log\log n$. Recall, that $c_4$ is the constant from  \cref{thm:CONGESTsplittingBoth}.  More formally, we set up the following $k$ bipartite splitting instances $B_i=(\VL_i\cup \VR_i, E_i)$, $i\in [k]$:
	$\VR_i=E(G_i)$ and $\VL_i=V(G_i)$. Note that the degree $d^{B_i}(v)=d^{G_i}(v)$ for a node $v\in \VR_i$ and $d^{B_i}(e)=2$ for a node $e\in \VR_i$. Hence, $B_i$ has maximum degree $\Delta'$. 
	
	We use \cref{thm:bipartiteSplitting} (for each $B_i$ in parallel and with the same $k'$ and $\eps'$) to compute a partition of $\VR_i$ into $\VR_{i,1},\ldots,\VR_{i,k'}$ such that each $v\in \VL_i$ has $d^{B_i}(v)/k'\pm\eps' \Delta'/k'=d^G(v)/(k\cdot k')\pm 3\eps'\Delta/(k\cdot k')$ neighbors in each $\VR_i$. Now, for $i\in [k], j\in[k']$ let $G_{i,j}=G_i[\VR_{i,j}]$ and note that  $G_{i,j}$ has maximum degree at most $\Delta''=\Delta'/k'+\eps'\Delta'/k'=\poly\log\log n$.
	
	In the last step, we apply \Cref{thm:edgeColoringSlow} on each $G_{i,j}, i\in [k], j\in [k']$ in parallel to edge-color $G_{i,j}$ with $(1+\eps/6)\Delta''$ colors in $\poly \Delta''+\poly\log\log n=\poly\log\log n$ rounds. 
	
	The total number of colors used is upper bounded by
	\begin{align*}
		k\cdot k'\cdot (1+\eps/6)\Delta''& 
		\leq k \cdot (1+\eps/6)^2\cdot \Delta'
		\leq (1+\eps/6)^3\cdot \Delta
		\leq (1+\eps)\Delta~.& & \qedhere
	\end{align*}	
\end{proof}

\subsection{Edge coloring in \CONGEST (almost identical to previous section)}
\label{ssec:CONGESTedgecoloring}

We begin with proving a \CONGEST counterpart of \Cref{thm:edgeColoringSlow} for very low degree graphs. 

\begin{theorem}
	\label{thm:CONGESTedgeColoringSlow}
	For any constant $\eps>0$, there is an absolute constant $\Delta_0$ such that there is a randomized \CONGEST algorithm with runtime $\poly\log\log n$ for $(1+\eps)\Delta$-edge coloring any $n$-node graph with maximum degree $\Delta_0\leq \Delta\leq \poly\log\log n$. 
\end{theorem}
\begin{proof}
	Recall from the text after \Cref{thm:edgeColoringSlow} that the edge-coloring problem can be solved via a constant number of LLL instances that are defined on a dependency graph $H$ such that one round of communication can be simulated in $O(1)$ rounds in \LOCAL in the original network \cite{CHLPU19}. The dependency degree of $H$ is $d=\poly\Delta$.  Hence, if $\Delta\leq \poly\log\log n$, one round of communication in $H$ can be simulated in $\poly\log\log n$ \CONGEST rounds in the communication network. The result follows via  \Cref{cor:LLLCONNGESTsmalldegree}. 
	The base case for the algorithm of \cite{CHLPU19} is a $5\Delta$-edge coloring step, which can also be solved in $\poly \log\log n$ \CONGEST rounds \cite{representativesets}.
\end{proof}

\begin{proof}[Proof of \Cref{thm:edgeColoring}, {\CONGEST}]
	We use the same high level algorithm as in the \LOCAL model, that is, we first split into subgraphs of $\Delta'=\poly\log n$ maximum degree, then in to subgraphs of $\Delta''=\poly\log\log n$ degree, which we then color with disjoint color spaces, with $(1+\eps/6)\Delta''$ colors each. We refer to the \LOCAL version for further the details on this reduction. Here, we only explain which parts differ in the \CONGEST model. 
	First, each edge is simulated by one of its endpoints. 
	The reduction to $\poly\log n$ degrees works in zero rounds, just as in the \LOCAL model. 
	
	The most challenging part is the reduction from $\poly\log n$ degrees to $\poly\log\log n$ degrees.  The pre-shattering phases (in computing the $q$-divide and in the main algorithm) immediately work in the \CONGEST model. We only need to reason that the post-shattering phases (of $q$-divide and the $q$ instances in the main algorithm)  can be solved via \Cref{lem:CONGESTpostshattering}. To this end, we need to run $\ell=O(\log n)$ instances of \cite{CPS14} in parallel with $\textsf{bandwidth}=\Theta(\log n)$. Observe that $k\leq \poly\log n$ holds. For each node to be able to evaluate the status of its bad events in all $\ell$ instances in parallel, we have the node simulating each edge send to the other endpoint the parts it chose in all $\ell$ instances (one part per instance). As the part index of an edge can be encoded with $O(\log k)$ bits, the indices in the $\ell$ instances can be communicated with $\ell\cdot O(\log k)=O(\log n \log\log n)$ bits. With the available bandwidth this requires only $O(\log\log n)$ rounds. The other messages needed by the algorithm to simulate the $\ell$ instances of \cite{CPS14} in parallel, such as whether the edge needs be re-sampled in each instance, similarly never exceed $O(\log \log n)$ rounds. 
	
	Once the degrees of the subgraphs are at most $\poly\log\log n$ we use \Cref{thm:CONGESTedgeColoringSlow} to color the graphs. Formally, the color space is still large. To really use \Cref{thm:CONGESTedgeColoringSlow}, vertices color each subgraph with colors in the range from $1$ to $\poly\log\log n$, and map their color back to the original color space at the end of the computation. 
\end{proof}

\section{Application: List Coloring}
\label{sec:listColoring}
We begin with defining the $(L,T)$-list-coloring problem. 
\begin{definition}[$(L,T)$-list-coloring]
	In an \emph{$(L,T)$-list-coloring} instance on a graph $G=(V,E)$, each node $v\in V$ is given a list $L(v)$ of colors of at least $L$ such that for each $c\in L(v)$ there are at most $T$ neighbors $u$ of $v$ with $c\in L(u)$. The parameter $T$ is referred to as the \emph{color degree}. Similarly for $c\in L(v)$, $|\{u\in N(v)\mid c\in L(u)\}|$ is the \emph{color degree} of color $c$ for $v$. 
\end{definition}
We first note that for the rest of the section, we assume that each vertex has a list of exactly $L$ colors, which, e.g., can be achieved if each node $v$ drops arbitrary $|L(v)|-L$ colors from its list. Further, we may assume that there is only an edge between two nodes if their lists intersect. Note, that still the degree of the graph $G$ can be much larger than the list size $L$; it can be as large as $L\cdot T$.  In particular, one cannot generally solve such the problem via a greedy approach, not even centrally. Still, the objective is to find solutions for arbitrary $L$ and $T$ with a ratio $L/T$ as small as possible. 

Reed \cite{Reed99} gave a simple LLL argument for the existence of a solution when $L/T \ge \lceil 2e \rceil$.  This was improved to $L/T = 2$ by Haxell \cite{Haxell01}. Reed and Sudakov \cite{RS02} then showed that $L/T = 1+o(1)$ suffices. Reed's famous list coloring conjecture states that $L=T+2$ colors always suffice \cite{Reed99}. We recall Reed's argument for the existence if $L/T>2e$.
In the distributed setting, there is an $O(\log n)$ round algorithm for $L/T\geq (1+\delta)$ \cite{CPS14}, and a $O(\poly\Delta+\poly\log\log n)$ rounds for $L/T\geq C_0$ for a sufficiently large constant $C_0$ \cite{FG17}.

\subparagraph{LLL formulation (for existence only).} Suppose each node picks a color from its list uniformly at random. Define a bad event $\mathcal{B}_{u,v,c}$ for each edge $\{u,v\}\in E$ and each color $c$ if both $u$ and $v$ choose the color $c$.
The probability for such an event is at most $p=1/|L(v)|\cdot 1/|L(u)| \le 1/L^2$. The dependency degree of these events is $d=2L \cdot T$, because  it can depend on at most $L$ colors for each of the endpoints of the edge and on $T$ other incident edges for each of these colors. 
Thus, we obtain the LLL criterion $p\cdot (2L\cdot T)=2T/L$, and hence for $L/T>2e$, a the standard criterion $epd<1$ is satisfied and a solution exists.

\subparagraph{Distributed results.} Reed's argument leads to an $O(\log^2 n)$-round  {\LOCAL} algorithm for any $L/T > 2e$ with the classic Moser-Tardos algorithm  \cite{MoserTardos10}.
Chang, Pettie and Su \cite{CPS14} gave a {\LOCAL} algorithm for $L/T=1+\delta$, with a quite involved analysis, that runs in time $O(\log^* L \max(1,\log n)/D^{(1-\gamma})$. 
Fischer and Ghaffari \cite{FG17} showed using \emph{color pruning} that there exists some (possibly large) constant $C$ to solve $(L,T)$-list coloring whenever $L/T\geq C$ in  $\poly(\Delta,\log\log n)$ rounds of {\LOCAL}.
\begin{theorem}[\cite{FG17}]
	\label{thm:distributedListColoringSlow}
	There is a constant $C$ and a $\poly(\Delta,\log\log n)$-round \LOCAL algorithm to solve any $(L,T)$-list coloring instance whenever $L/T\geq C$. 
\end{theorem}

In this section, we combine the effectiveness of \cite{CPS14} with the speed of \cite{FG17}. In particular, we prove the following result.

\thmListColoring*

Similarly to the edge-coloring problem our high level idea is to first reduce the size of the relevant parameters  to $\poly\log\log n$, after which we can solve arbitrary LLLs efficiently on the problem. However, the reduction and the base case (once parameters are of size $\poly\log\log n$)  
are significantly more involved than in the edge-coloring problem. We begin with a statement on efficiently reducing the parameters $L$ and $T$ while keeping the ratio of list size and color degree almost the same.

\begin{lemma}[List color sparsification]
	\label{lem:listColorReduction}
	There exists a universal constant $\inccte>0$ s.t.:
	For any $\eps>1/\poly\log\log n$ and
	$k \leq \cte \cdot(\eps^4L /\log L)$,
	there is a $\poly(\eps^{-1},\log\log n)$-round algorithm for the following list coloring sparsification problem: Given a $(L,T)$-list coloring instance with $T<L\leq \poly\log n$ on an  $n$-node graph $G=(V,E)$ the goal is to compute a sublist $L'(v)\subseteq L(v)$ for each node yielding a $(L',T')$-list coloring instance on the same graph with  
	\begin{align}
		L'= L/k\pm \eps L/k, 
		T'\leq T/k+\eps T/k \text{ and }
		L'/T'\geq (1-\eps)L/T. 
	\end{align}
	
	We obtain the same properties in zero rounds if $L>\poly\log n$  and $k\leq \eps^2\Delta/(9\ln n)$ holds for $\Delta=L\cdot T$.
\end{lemma}
\begin{proof}
	Consider the $(L,T)$-list coloring instance on graph $G=(V,E)$ with a list $L(v)$ for each node $v\in V$. 
	
	For the zero round statement, every node just keeps each of its colors independently with probability $1/k$. In expectation, the list size will be $\Delta/k\geq \poly\log n$, and in expectation the color degree also shrinks by a factor $k$. Both values are w.h.p.\ concentrated with discrepancy $\pm\eps L/k$ and $\pm\eps T/k$ with the same Chernoff bound reasoning as in \Cref{obs:zeroRoundSplitting}. 
	
	For the case that $L\leq \poly\log n$ we perform the following reduction. 
	
	\begin{figure}[ht]
		\centering
		\hspace{0.05\textwidth}
		\begin{subfigure}{0.2\textwidth}
			\includegraphics[page=11,width=\textwidth]{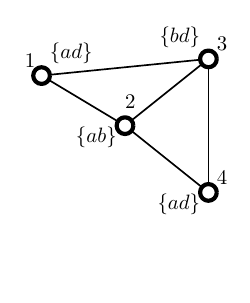}
		\end{subfigure}
		\hfill
		\begin{subfigure}{0.2\textwidth}
			\includegraphics[page=12,width=\textwidth]{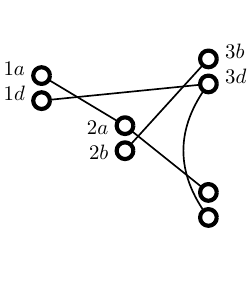}
		\end{subfigure}
		\hfill
		\begin{subfigure}{0.2\textwidth}
			\includegraphics[page=13,width=\textwidth]{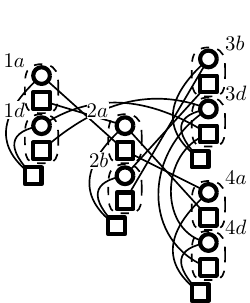}
		\end{subfigure}
		\hspace{0.05\textwidth}
		\caption{A graph (left), the associated graph representing the conflicts between palette colors (middle), and the bipartite splitting instance obtained from the list coloring problem on this graph (right).}
		\label{fig:bipartiteTranslateList}
	\end{figure}
	
	\textbf{Bipartite splitting instance:} Create the a bipartite splitting instance $B=(\VL\cup \VR, E_\Bmark)$: 
	\begin{itemize}
		\item Vertices: $\VR=\{(v_{c} \mid v\in V, c\in L(v)\}$ and $\VL=\{\mathcal{B}_v \mid v\in V\}\cup \{\mathcal{B}_{v,c}\mid v\in V, c\in L(v)\}$,
		\item Edges: For each $v\in V, c\in L(v)$, $\mathcal{B}_v$ and $v_c$ are linked; For each $\{u,v\}\in E$ and $c\in L(u)\cap L(v)$,  $\mathcal{B}_{v,c}$ and $u_c$ are linked, as are $\mathcal{B}_{u,c}$ and $v_c$.
	\end{itemize}
	Note that the nodes $\mathcal{B}_v\in \VL$ have degree at most $\DeltaL=L$, the nodes $\mathcal{B}_{v,c}\in \VL$ have degree at most $T$, and the nodes $v_c\in \VR$ have degree $T+1$. Since $L>T$, the maximum degree $\Delta_B$ of the graph $B$ is upper bounded by $L$. Thus, with $c_6$ chosen such that $c_6 \leq c_2$, a $k$ which satisfies the hypotheses of \cref{lem:listColorReduction} also satisfies those of \cref{thm:bipartiteSplitting}, i.e., $k\leq c_2\eps^4 \DeltaL/\log \Delta_B$.

	Now, apply \Cref{thm:bipartiteSplitting} with $\eps/100$ and $k$ and let $\VR_1,\ldots,\VR_k$ be the returned partition of $\VR$ (we will only make use of $\VR_1$). The runtime of \Cref{thm:bipartiteSplitting} is $\poly\log\log n$ rounds.

	\textbf{New list-coloring instance:} The graph remains $G=(V,E)$, but the list $L'(v)$ of a node $v$ becomes $\{c \mid v_c\in \VR_1\}$. The size of $L'(v)$ equals the degree into $\VR_1$ of  event node $\mathcal{B}_{v}$. Thus, we obtain that  $|L'(v)|=L/k\pm \eps \DeltaL/k$. Similarly, due to the events $\mathcal{B}_{v,c}$ we can bound the color degree for each color $c$ by $T'(v)=T/k\pm \eps \DeltaL/k$. 
\end{proof}

\begin{corollary}
	There exists a constant $C$ and a $\poly \log\log n$-round algorithm that solves the $(L,T)$ list coloring problem whenever $L/T\geq C$.
	\label{cor:listcolLargeConstant}
\end{corollary}
\begin{proof}
	Let $C'$ be the constant in \cref{thm:distributedListColoringSlow}. 
	We first run the algorithm of  \Cref{lem:listColorReduction} (possibly we run it twice, first to reduce to $L\leq \poly\log n$ in zero rounds, and then a second time to reduce to $L\leq \poly\log\log n$) to compute a new $(L',T')$-list coloring instance, where $L'/T'$ is at least $C'$ and both $L'$ and $T'$ are bounded by $\poly\log\log n$. We then solve the resulting $(L',T')$-list coloring instance (and thus also the original $(L,T)$-list coloring problem) in $\poly(L'\cdot T',\poly\log\log n)=\poly\log\log n$ rounds via  the algorithm of  \cref{thm:distributedListColoringSlow}. 
	
	If initially $d = \Omega(\log n)$ we require two iterations of \Cref{lem:listColorReduction} and since we  lose a small constant factor in the ratio of the list size and the color degree in each iteration we need $L/T\geq C$ for some constant $C>C'$. 
\end{proof}

\subsection{List Coloring Constant Amplification}
Our second method is based on the nibble method of Reed-Sudakov \cite{RS02}. In order to show that the list coloring constant asymptotically equals one, Reed and Sudakov designed a process that (slowly) amplifies the constant of the given list coloring problem. They convert a list coloring problem with parameters $L=(1+\delta)g$ and $T=g$ to one with $L' = g'$ and $T'=g'/q$, where $q$ is any constant and $g' = \Theta(g)$. A solution to the second problem immediately solves the first one. In fact, their process consists of $O(\ln g)$ iterations, in each of which the ratio between the list size $L_{i+1}$ and the color degree $T_{i+1}$ improves by a factor $\approx 1+\Theta(1/(1+\delta)\ln g)$ compared to $L_i/T_i$. In creating the  $(L_{i+1},T_{i+1})$-list coloring instance from the $(L_i,T_i)$ one, some of the nodes are colored while others have their palettes pruned.
The crucial property is to ensure that lists do not shrink too fast, i.e., bounding $L_{i+1}$ from below, and to ensure that the color degree $T_{i+1}$ shrinks fast enough, i.e., bounding it from above. As they were only interested in the existence of colorings, it was sufficient to obtain the desired bounds with a positive probability only.  However, it turns out that one step of their process is an LLL. Next, we describe the process.

\subparagraph{Reed-Sudakov LLL process \cite{RS02}.} Set $L_0(v)=L(v)$.  We next explain iteration $i$ of the process. Each vertex is either already colored or has a given list $L_i(v)\subseteq L(v)$. Now, each uncolored vertex activates itself with probability $p=1/\ln g$ and picks a color $c_v$ from its list $L_i(v)$ uniformly at random. The color $c_v$ is removed from all lists of its neighbors  and $v$ is permanently colored with $c_v$ if no neighbor picked the same color. For each uncolored vertex $v$, the list $L_{i+1}(v)$ in iteration $i+1$ consists of $L_i(v)$ without the removed colors.

Let $T_i(v)$ be the maximum color degree of a color of node $v$ in the $i$-th instance. 
\begin{lemma}[\cite{RS02}]
	\label{lem:listColLLL}
	If $|L_i(v)|\geq (1+\delta)T_i(v)$ for all nodes $v\in V$, then the above process forms an LLL with error probability $g^{-\ln g}$, dependency degree $g^{8}$, and with the following \emph{bad events}
	\begin{enumerate}
		\item $v$ is uncolored and $|L_{i+1}(v)|< \left(1-\frac{1}{(1+3\delta/4)\cdot \ln g}\right)|L_i(v)|$~,
		\item $v$ is uncolored and $|T_{i+1}(v)|> \left(1-\frac{1}{(1+\delta/4)\cdot \ln g}\right)|T_i(v)|$~.
	\end{enumerate}
\end{lemma}
\begin{proof}
	In \cite[Lemma 3.2]{RS02}, Reed and Sudakov bound the probability of the bad events by $g^{-\ln g}$. 
	
	Both events are independent from all randomness that is further than $4$-hops apart and the degree of $G$ is upper bounded by $L\cdot T=O(g^2)$ (as we can discard any edges between nodes with disjoint lists). Hence, the dependency degree of the LLL is at most $g^8$.
\end{proof}

In order to obtain the constant from \Cref{thm:distributedListColoringSlow}, in the next lemma, we show that repeating the process $O(\log C\log g)$ times increases the list vs.\ color degree ratio by a factor $C$, while the minimum list size reduces by a constant factor only.

Let $L_i$ and $T_i$ be the globally smallest list size and largest color degree, respectively.
\begin{lemma}[$L_i/T_i$ improves, lists remain large]
	\label{lem:listColBounds}
	For any constant $C$, and after iterating the LLL process for $r=\lceil (5/\delta)\ln C \ln g \rceil$ iterations such that the respective bad events in \Cref{lem:listColLLL} are avoided after each iteration,  we have $L_r/T_r\geq C$ and we also have the following lower bound on the list size of any uncolored node $L_r\geq\left(1-\frac{1}{(1+3\delta/4)\cdot \ln g}\right)^rL\geq  C^{-\Theta(\delta^{-1})}L$. 
\end{lemma}
\begin{proof}
	Recall that $1-x \leq e^{-x} \leq 1-x/2$ for all $x\in [0,1]$. Similar, to \cite[Proof of Theorem 1.1]{RS02} we obtain $L_r/T_r\geq C$ because
	\begin{align*}
		\frac{T_r}{L_r}
		\leq \frac{\left(1-\frac{1}{(1+\delta/4)\cdot \ln g}\right)^r|T|}{\left(1-\frac{1}{(1+3\delta/4)\cdot \ln g}\right)^r|L|}
		\leq \frac{\left(1-\frac{1}{(1+\delta/4)\cdot \ln g}\right)^r}{\left(1-\frac{1}{(1+3\delta/4)\cdot \ln g}\right)^r}\leq \left(1-\frac{\delta}{5\ln g}\right)^r\leq 1/C~.
	\end{align*}
	
	The second part of the claim follows due to 
	\begin{align*}
		L_r\geq\left(1-\frac{1}{(1+3\delta/4)\cdot \ln g}\right)^r L\geq e^{-1-(10/\delta) \ln C /(1+3\delta/4)}L\geq
		(1/e)C^{-10/\delta} L~. & \qedhere
	\end{align*}
\end{proof}

Next, we show that we can also efficiently increase the ratio in the \LOCAL model.
\begin{lemma}[Increasing $L/T$]
	\label{lem:listColIncreaseRatio}
	For any constants $q, \delta$, and $d=\omega(1)$, there is a $\poly(\Delta) + O(\log \Delta)\cdot \poly(\log\log n)$-round \LOCAL algorithm that reduces a list coloring problem with parameters $L=(1+\delta)g$ and $T=g$ to one with $L' = g'$ and $T'=g'/q$, where $g' = \Theta(g)$.
\end{lemma}
\begin{proof}
	The first claim follows from the LLL formulation of \Cref{lem:listColLLL}, the $\poly\Delta+\poly\log\log n$ LLL algorithm by Fischer and Ghaffari \cite{FG17,RG19} and the bounds on list size $L_i$ and the $L_i/T_i$ ratio from \Cref{lem:listColBounds}. The $O(\log \Delta)$ term in the runtime stems from the $O(\log g)=O(\log \Delta)$ iterations of amplifying the constant. Due to $g=\omega(1)$ none of the lists becomes empty unless a node is colored in the process. 
\end{proof}

\begin{proof}[Proof of \cref{thm:ListColoring}] For some constant\footnote{With regard to $\delta$, the bottleneck of the procedure is the decrease in the list size captures in \Cref{lem:listColBounds}. 
		When invoking the lemma lists are of size $\poly\log\log n$ and the resulting lists are by a $2^{\Theta(1/\delta)}$ factor smaller, but are not allowed to vanish. This is the case if $\delta=\omega(1/\poly\log\log\log n)$.} $\delta>0$, we begin with a list coloring instance with $L/T>(1+\delta)$. 
	We first run two recursive iterations of \Cref{lem:listColorReduction} to compute sublists $L''(v)\subseteq L'(v)\subset L(v)$, such that $|L''(v)|\geq L''=\poly\log \log n$ and the resulting list-coloring instance has the color degree upper bounded by $T''=\poly\log\log n$, while at the same time $L''/T''>(1+\delta/10)$.  Then, let $C$ be the constant of \Cref{thm:distributedListColoringSlow} and apply \Cref{lem:listColIncreaseRatio} on the $(L'',T'')$-list coloring instance until we have achieved list size vs. color degree ratio of $>C$ (the degree of the resulting instance is $\Delta'=C^{-\Theta(1/\delta)}\cdot \poly\log\log n$), and finally use the algorithm of \Cref{thm:distributedListColoringSlow} to solve the remaining instance in $\poly\Delta'+ \poly\log\log n=\poly\log\log n$ rounds. 
\end{proof}

\section{Bipartite Vertex Splitting and Beyond}
\label{sec:bipartiteSplitting}
Another classic version is to split only one side of a bipartite graph.
Given a bipartite graph $(\VL\cup \VR, E)$ and an parameter $k$ the objective is to split the \emph{variable vertices} $\VR$ into $k$ parts $\VR_1,\ldots, \VR_k$ such that the degree of every \emph{event vertex} $u\in \VL$ into each part $\VR_i$ does not deviate  from $d(u)/k$ by too much. More formally, each event node $u\in \VL$ comes with a parameter $z(u)$ that bounds the deviation. Let $\DeltaL$ and $\DeltaR$ be the maximum degree of nodes in $\VL$ and $\VR$ respectively. With the same analysis  as for \Cref{thm:vertexSplitting} (reasons below) we obtain the following theorem for bipartite vertex splitting.

\thmbipartiteVertexSplitting*

The simpler $q$-divide problem also naturally extends to this more general setup. The objective of a bipartite $q$-divide is to partition the variable vertices into $q$ parts such that each event node has at most $8\DeltaL/q$ neighbors in each part.

\begin{theorem}
	\label{thm:bipartiteDivide}
	For any $q \in [1,(1/6)\DeltaL/\ln \Delta]$, there is a \LOCAL algorithm to compute a bipartite $q$-divide in $\poly \log \log n$-rounds.
\end{theorem}

\begin{figure}[ht]
	\centering
	\hspace{0.05\textwidth}
	\begin{subfigure}{0.2\textwidth}
		\includegraphics[page=1,width=\textwidth]{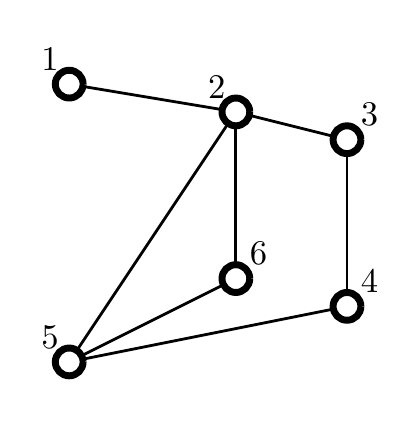}
	\end{subfigure}
	\hfill
	\begin{subfigure}{0.2\textwidth}
		\includegraphics[page=2,width=\textwidth]{splitting}
	\end{subfigure}
	\hfill
	\begin{subfigure}{0.2\textwidth}
		\includegraphics[page=3,width=\textwidth]{splitting}
	\end{subfigure}
	\hspace{0.05\textwidth}
	\caption{A graph and the bipartite splitting instance obtained from the vertex-splitting problem on it by turning each node into a variable node (circles) and an event node (squares).}
	\label{fig:bipartiteTranslate}
\end{figure}

To better understand this more general setting and how our results for $q$-divide and $k$-split extend to it, let us translate those problems into their bipartite versions. For a graph $G=(V,E)$, we construct a bipartite graph $G'=(\VL \cup \VR,E')$ such that a bipartite $k$-split (bipartite $q$-divide) on $G'$ maps to a $k$-split ($q$-divide) on $G$. Let $n=\card{V}$ be the number of nodes of $G$ and $\Delta$ its maximum degree. In bipartite terminology, when computing a $k$-split on $G$ each node of $G$ is acting both as an event node and a variable node, as we must ensure the proper splitting of its neighborhood as well as assigning it. The bipartite graph corresponding to this problem is the graph $G'=(\VL \cup \VR,E')$ where $\card{\VL} = \card{\VR} = n$, $\forall i,j \in \card{V}^2, v_i v_j \in E \Leftrightarrow \parens*{v^\Lmark_i v^\Rmark_j \in E' \wedge v^\Rmark_i v^\Lmark_j \in E'}$. $G'$ has $2n$ nodes and maximum left and right degree $\DeltaL = \DeltaR = \Delta$. See \cref{fig:bipartiteTranslate} for an illustration of the translation process.

\begin{proof}[Proof of \Cref{thm:bipartiteSplitting,thm:bipartiteDivide}]
	Our proofs of \cref{thm:vertexSplitting,thm:qdivide} naturally extend to the bipartite setting. Intuitively, the algorithms follow the same pattern. We have a simple random assignment procedure that we show properly partitions the neighborhood of a node w.p.\ $1-\poly(\Delta)$. In addition, there is a way of running this procedure, retracting some assignments and avoiding to assign some nodes that ensures that only small patches of the graph remain unassigned and the partial assignment that is obtained can be completed to a full assignment. All that we need to show is that our setting of parameters in the bipartite setting are correct, i.e., control the amount of discrepancy as in the previous setting.
	
	Let $\DeltaL$ be the degree of the left hand side vertices in $V^{\mathcal{E}}$ and let $\DeltaR$ be the degree of the right hand side vertices in $V$. Let $\Delta=\max\{\DeltaL,\DeltaR\}$. 
	
	All probabilities are exponential in $-\Theta(z)$, or in $\Theta(-z \cdot (z/(d(u)/k)))$ if the degree of a node $u$ is larger than $k\cdot z$. 
	
	The union bound in the shattering over distance $5$-neighborhoods introduces a multiplicative $\Delta^5$ term. Hence, we require that $e^{\Theta(z)}$ and $e^{\Theta(z)\cdot(z/(d(u)/k))}$ dominate the $\Delta^5$ term. This, clearly holds if $z=\eps^2\Delta/(72k)$ as before, given the assumed upper bound on $k$. 
	
	The proof of the discrepancy (in \cref{s:discAnalysis}) remains exactly the same; just note that in the bipartite vertex splitting the discrepancy values ($z(v)$s) depend on $\DeltaL$ instead of $\Delta$, and hence we obtain a deviation from $d(u)/k$ that is upper bounded by $\eps\DeltaL/k$. 
\end{proof}

\begin{remark}
	In general, it is not possible to recursively use vertex-splitting to split into smaller and smaller parts. Special properties of an instance (as we have with edge-splitting and when solving list-coloring here) sometime still make it possible.
\end{remark}
Intuitively, one would like to compute a vertex-splitting recursively. For a large $\Delta\gg \poly\log n$, one would first like to set $k=O(\Delta/\poly \log n)$ in order to split into subgraphs of maximum degree $\Delta'=O(\Delta/k)=\poly\log n$, only to then apply \Cref{thm:vertexSplitting} for the case of polylogarithmic degrees to split further into graphs of $\poly\log\log n$ degree. However, this approach does not work, which can best be illustrated in the (equivalent) bipartite formulation of the problem. Initially, we have $\DeltaL=\DeltaR\gg\poly\log n$. The first iteration of the splitting now reduces only the degree of the event nodes on the left hand side, that is, we reduce $\DeltaL$ to some $\DeltaL'=\poly\log n$ while the degree of the variable nodes remains unchanged. Thus, the new instance still has the original maximum degree of $\Delta$ and we did not gain any degree reduction in order to apply \Cref{thm:vertexSplitting} on a smaller degree graph.

It is however possible to recurse if $\DeltaR$ is small from the beginning or if we do not just split and can argue that it goes down. This is the case in our applications. $\DeltaR=2$ in the case of edge-splitting. In list coloring, the argument is finer. We only recurse on $\VR_1$ and argue that some event nodes are directly associated to a variable node and can be ignored if this variable node is not in the recursion. This is in spirit similar to what happens in defective coloring, where each node $v$ of the original graph acts both as an event node $v^{\mathcal{E}}$ and as a variable node $v^{\mathcal{V}}$, and the event node $v^{\mathcal{E}}$ becomes irrelevant to variable nodes that chose a color distinct from that chosen by $v^{\mathcal{V}}$.

\section{Local \texorpdfstring{$q$}{q}-divide}
\label{ssec:qdivideLocal}
In \Cref{thm:qdivide}, the number of nodes in a bucket is upper bounded by $8\Delta/q$, regardless of the degree of the node. In the next theorem we extend the result to a more fine-grained bound that depends on a node's degree. To this end let a \emph{local $q$-divide} be a partition of the vertices into $q$ parts such that each node has at most $\max\{8d(v)/q, 48\ln \Delta\}$ neighbors in each of the $q$ parts.
\begin{restatable}{theorem}{thmqdividelocalGuarantees}
	\label{thm:qdivideLocalGuarantees}
	For any $q\in [1,(1/6)\Delta/\ln\Delta]$ there is a $\poly\log\log n$-round \LOCAL algorithms to compute a local $q$-divide.
\end{restatable}
\begin{proof}
	We use the same algorithm as for \Cref{thm:qdivide}, but with a customized threshold for each node, i.e., we set  $z(v)=8d(v)/q$ for all nodes with $d(v)\geq 6 q\ln \Delta$ and $z(v)=48\ln \Delta$ otherwise. For this choice of $z(v)$ \Cref{lem:qdivideShattering} and \Cref{lem:qdividePostShatteringLLL} still hold, respectively: The only parts in the proofs that change are the probability bounds on the events that the number of neighbors $D_j$ in bucket $j$ exceeds $z(v)$ for a vertex $v$ with degree $d(v)$. Set  $\delta=(z(v)-\Exp[D_j])/\Exp(D_j)$ and observe that $\delta\geq 1$ in all cases. 
	We obtain with \Cref{lem:basicchernoff} (Chernoff bound)
	\begin{align*}
		\Pr(D_j\geq z(v)) 
		& \leq \Pr(|D_j-E[D_j]|\geq \delta \Exp[D_j]
		\\ & \leq \exp(-\delta\Exp[D_j]/3 )
		\leq \exp(-(z(v)-\Exp[D_j])/3)~. 
	\end{align*}
	If $z(v)=8d(v)/q$ (and consequently $d(v)\geq 6q \ln \Delta$ as well as $\Exp[D_j]\geq 12\ln \Delta$), we have $z(v)=4\Exp[D_j]$ and the bound simplifies to $\exp(-\Exp[D_j])\leq  \exp(-12\ln \Delta)\leq \Delta^{-12}$.
	For $z(v)=48\ln \Delta$ (and consequently $d(v)\leq 6q\ln \Delta$) we obtain $\Exp[D_j]\leq 12\ln \Delta$. Hence, the respective probability is upper bounded by $\exp(-(z(v)-\Exp[D_j])/3)\leq \exp(-12\ln \Delta)\leq \Delta^{-12}$.
\end{proof}

\newpage

\appendix

\section{Efficient Postshattering in \CONGEST for Small Degrees}
\label{app:postshattering}
The proof of \Cref{lem:CONGESTpostshattering} is similar to an \CONGEST LLL algorithm in \cite{MU21}  for instances of size $N=O(\log n)$ and the case of $d=O(1)$. In fact, following all dependencies on $d$ (and a slightly increased $N$) in the proof of \cite{MU21} yields an algorithm with runtime $\poly (d, \log\log n)$, which yields the desired runtime whenever $d=\poly\log \log n$. 
\begin{proof}[Proof sketch of \Cref{lem:CONGESTpostshattering}]
	We adapt the proof for a similar statement in \cite{MU21}. As the corresponding statement only considers the case of $\Delta=O(1)$, they hide all factors depending on $\Delta$ in the $O$-notation. But, in fact, the proof has a polynomial dependency on $\Delta$ in the runtime. Thus, its runtime remains $\poly\log \log n$, if  $\Delta=\poly\log\log n$.
	
	Next, we sketch their algorithm and highlight the dependency on $\Delta$ in the runtime: From a high level point of view they use he method of conditional expectation to derandomize the LLL algorithm of \cite{CPS14}. This algorithm runs in $T=O(\log N)=O(\log \log n)$ rounds and has error probability $1/N^2$. Thus, initially, if all random bits of all nodes were chosen u.a.r., in expectation one would have $\E[\sum_{u} F_u<1]$ where $F_u$ is the indicator variable that the algorithm fails to avoid bad event $u$. 
	
	The derandomization in \cite{MU21}: Compute a $\Omega(\log\log n)$-separated network decomposition with $Q=O(\log\log n)$ color classes via \Cref{thm:networkDecomp}. Then, iterate through the color classes and derandomize the \cite{CPS14} algorithm for solving the respective (polynomial) LLLs by finding good random bits for the nodes in each cluster. At the end, the \cite{CPS14} algorithm is executed with these random bits. The crucial invariant that is maintained is $\E[\sum_{u\in U} F_u<1 \mid X]$ where $X$ describes how already considered random bits are fixed.  At the end all random bits are set, and in expectation the (now deterministic) algorithm does not fail, that is, it computes a valid solution avoiding all bad events. The crucial ingredient for the runtime is a full information gathering when processing a cluster (clusters of the same color class of the network decomposition can be handled independently and in parallel due to their large enough distance $\Omega(\log\log n)$), that is, the whole topology of the cluster and its distance $(T+r)$-hop neighborhood in the small component, where $T(N)=O(\log \log n)$ is the runtime of \cite{CPS14} on the small component and $r=O(1)$ is the checking radius of the bad events. As the component size is upper bounded by $N$, and each node has maximum degree $\Delta$ and the \cite{CPS14} algorithm uses only $\poly\log N=\poly\log\log n$ random bits per node, this information gathering can be done by sending only $N\cdot \Delta$ ID pairs to encode the graph topology and $N\cdot \poly\log \log n$ values of already set random bits. After temporarily renaming all nodes in the cluster with IDs from range $[N]$ (a simple BFS is sufficient to do this), all this information can be learned in $O(D+\max\{1,N\poly\log \log n/\textsf{bandwidth})=\poly\log \log n$ rounds, using the pipelining result from \Cref{cor:treeAggregationBetter}. Here $D$ is the (weak) cluster diameter, and $\textsf{bandwidth}=\Omega(\log n)$ is the bandwidth of the \CONGEST model. 
\end{proof}

\section{Communication on Top of Weak Network Decompositions}
\label{app:networkDecompositionRouting}
\begin{lemma}[\cite{GGR20,MU21}]
	\label{cor:treeAggregationBetter}
	Let $G$ be a communication graph on $n$ vertices. Suppose that each vertex of $G$ is
	part of some cluster $\mathcal{C}$ such that each such cluster has a rooted Steiner tree $T_{\mathcal{C}}$ of diameter at most
	$\beta$ and each node of $G$ is contained in at most $\kappa$ such trees. Then, in $O(\max\{1,\kappa/b \}\cdot (\beta + \kappa))$ rounds of the
	\CONGEST model with $b$-bit messages, we can perform the following operations for all
	clusters in parallel on all clusters:
	\begin{compactenum}
		\item  Broadcast: The root of $T_{\mathcal{C}}$ sends a $b$-bit message to all nodes in $\mathcal{C}$;
		\item  Convergecast: We have $O(1)$ special nodes $u \in \mathcal{C}$, where each special node starts with a
		separate $b$-bit message. At the end, the root of $T_{\mathcal{C}}$ knows all messages;
		\item  Minimum: Each node $u \in \mathcal{C}$ starts with a non negative $b$-bit number $x_u$. At the end, the root
		of $T_{\mathcal{C}}$ knows the value of $\min_{u \in \mathcal{C}}x_u$;
		\item  Summation: Each node $u \in \mathcal{C}$ starts with a non negative $b$-bit number $x_u$. At the end, the root
		of $T_{\mathcal{C}}$ knows the value of $\big(\sum_{u\in \mathcal{C}} x_u\big) \mod 2 ^{O(b)}$.
	\end{compactenum}
\end{lemma}

\bibliographystyle{plain}
\bibliography{refs}

\end{document}